\documentclass[pra,aps,twocolumn,superscriptaddress,longbibliography]{revtex4-1}
\usepackage[colorlinks=true, citecolor=red, urlcolor=blue ]{hyperref}
\usepackage{graphicx}
\usepackage{bm}
\usepackage{amsmath,amsfonts}
\usepackage{amsthm}
\usepackage{hyperref}
\usepackage{xcolor}
\hypersetup{
    urlcolor=blue,
    citecolor=blue
           }

\usepackage{braket}
\theoremstyle{plain}
\usepackage{color}
\usepackage{amssymb}
\usepackage{amsthm}
\usepackage{amsfonts}
\usepackage{float}
\usepackage{tabularx}
\usepackage{graphicx}
\usepackage[export]{adjustbox}

\usepackage{mathtools}
\usepackage{esvect}
\usepackage{wrapfig}
\usepackage{amsthm}
\usepackage{verbatim}
\usepackage{bbm}
\usepackage[normalem]{ulem}

\usepackage{enumitem}
\usepackage{fmtcount}
\usepackage{booktabs}
\usepackage{csquotes}
\usepackage{epsfig}

\usepackage{tabularx}
\usepackage{graphicx}
\usepackage{amsmath}
\usepackage{braket}
\usepackage{latexsym}
\usepackage{bm}
\usepackage{graphics,epstopdf}
\usepackage{enumitem}
\usepackage{fmtcount}
\usepackage{booktabs}
\usepackage{csquotes}
\usepackage{epsfig}

\theoremstyle{plain}

\def\bea{\begin{eqnarray}}
\def\eea{\end{eqnarray}}
\def\ba{\begin{array}}
\def\ea{\end{array}}

\def\beq{\begin{equation}}
\def\eeq{\end{equation}}

\usepackage[normalem]{ulem}
\usepackage{float}
\usepackage{graphicx}  
\usepackage{dcolumn}          
\usepackage{amssymb}
\usepackage{appendix}
\usepackage{physics}   
\usepackage{mathtools}
\usepackage{esvect}
\usepackage{wrapfig}
\usepackage{amsthm}
\usepackage{verbatim}
\usepackage{bbm}

\usepackage[mathscr]{euscript}
\def\Tr{\operatorname{Tr}}

\def\({\left(}
\def\){\right)}
\def\[{\left[}
\def\]{\right]}

\newtheorem{theorem}{Theorem}

\newtheorem{corollary}{Corollary}

\begin{document}
\title{\textbf {Minimal-error quantum state discrimination versus robustness of entanglement:\\More indistinguishability with less entanglement}}
\author{Debarupa Saha}
\affiliation{Harish-Chandra Research Institute,  A CI of Homi Bhabha National Institute, Chhatnag Road, Jhunsi, Prayagraj  211019, India}
\author{Kornikar Sen}
\affiliation{Harish-Chandra Research Institute,  A CI of Homi Bhabha National Institute, Chhatnag Road, Jhunsi, Prayagraj  211019, India}
\author{Chirag Srivastava}
\affiliation{Laboratoire d'Information Quantique, CP 225, Université libre de Bruxelles (ULB), Av. F. D. Roosevelt 50, 1050 Bruxelles, Belgium}
\author{Ujjwal Sen}
\affiliation{Harish-Chandra Research Institute,  A CI of Homi Bhabha National Institute, Chhatnag Road, Jhunsi, Prayagraj  211019, India}
\begin{abstract}
We relate the the distinguishability of quantum states with their robustness of the entanglement, where the robustness of any resource quantifies how tolerant it is to noise. In particular, we identify upper and lower bounds on the probability of discriminating the states, appearing in an arbitrary multiparty ensemble, in terms of their robustness of entanglement and the probability of discriminating states of the closest separable ensemble. 
These bounds hold true, irrespective of the dimension of the constituent systems the number of parties involved, the size of the ensemble, and whether the measurement strategies are local or global. Additional lower bounds on the same quantity is determined by considering two special cases of two-state multiparty ensembles, either having equal entanglement or at least one of them being separable. The case of equal entanglement reveals that it is always easier to discriminate the entangled states than the ones in the corresponding closest separable ensemble, a phenomenon which we refer as ``More indistinguishability with less entanglement".
Furthermore, we numerically explore how tight the bounds are by examining the global discrimination probability of states selected from Haar-uniformly generated ensembles of two two-qubit states. We find that for two-element ensembles of unequal entanglements, the minimum of the two entanglements
must possess a threshold value for the ensemble to exhibit ``More indistinguishability with less entanglement".
\end{abstract}
\maketitle
\section{Introduction}  
The task of identifying a quantum state, chosen from a finite set of states with a certain probability, by performing measurements on the state is known as quantum state discrimination (QSD). QSD has various applications, such as, 
 channel discrimination ~\cite{CD1,CD2}, random access code~\cite{Rac1,Rac2}, dense coding~\cite{dc1,dc2,d3,d4,dc5,dc6} and metrology~\cite{MR1,MR0,MR2,MR3,MR4,MR5}. QSD therefore has technological as well as foundational importance.

The no-cloning theorem~\cite{I1} prohibits perfect discrimination of non-orthogonal states using quantum operations. Even if the state is selected from a set of mutually orthogonal states, the state may not be perfectly distinguishable if it is shared between multiple parties and only local operations and classical communication is allowed. State discrimination can broadly be categorized into two classes: minimum-error state discrimination (MESD)~\cite{Me0,Hels,Me1,Me2,Me3,Me4,Me5,Me6,Me7,Me8,Me10,Me9,Me11,Me12,Me13,Me14,Me15}  and unambiguous state discrimination (USD)~\cite{U2,U3,U1,U4,U5,U7,U6,U8,U9,U11,U12,U13,U14,U15,U10,U16,U17,U18}. In MESD, there is always a guess made for every measurement outcome, but there is a certain probability that the inference is incorrect. On the other hand, in USD, the state discrimination may not hand out a guess corresponding to an outcome, but when one is available, it is certainly correct. See~\cite{Re1,Re2,Re3,Re4,Re5,Re6,Re7,Re8,Re9} for reviews.

QSD has been experimentally realised, often using optical systems~\cite{E1,E2,E3,E5,E6,E4,E7}. 
In Ref.~\cite{E8}, the authors have employed quantum dots for multi-state discrimination. Some experimental setups deal with discrimination of non-orthogonal coherent states, as can be found in~\cite{E9,E10,E11,E12,E13}. Moreover, there are certain experiments which have used multiple copies of the state for better discrimination Refs.~\cite{Mc,Mc2,Mc4}.


An independent challenge of MESD lies in finding an optimal measurement strategy that will maximize the probability of correct guessing. Though the optimal global measurements for discrimination of two arbitrary but fixed quantum states is known~\cite{Hels}, for more than two states no general form of the optimal measurement is available. See~\cite{Up1,Up2,Up3,Up4,Up5,Up6,Up7,Up8,Up9,Upb9,Up10,Upb11,Up11}, for various special cases.
Ref. ~\cite{UL} presented bounds on the optimal probability of correct identification of the states with local operations and classical communication. But this bound may not always be attainable.

The role of different quantum resources in QSD tasks is as yet unclear. In~\cite{Co}, the authors related the problem of QSD to ``robustness’’ of the measurements involved in the distinguishing. A relation of QSD with quantum coherence was found in~\cite{Cb1, Cb2}. Here we show that the optimal probability of inferring a state from an arbitrary multisite ensemble in an arbitrary MESD protocol is related to the entanglement content of the elements of the ensemble. More specifically,  we introduce a notion of the closest separable ensemble for an arbitrary multiparty ensemble, and show that the ratio of the optimal guessing probability in an arbitrary MESD protocol for an arbitrary multiparty ensemble to the same for its closest separable ensemble is upper bounded by a measure of entanglement content of the constituent states of the ensemble. A different measure of entanglement of the states is found to lower bound a related quantity. Both the measures are functions related to the random robustness of entanglement ~\cite{Ref1}. 

We subsequently consider the case of global discrimination of two-element multiparty ensembles. We find an independent lower bound on the probability of distinguishing the states in two special instances, viz. the case when the two states of the ensemble have the same entanglement, and the case when one of them is separable. The entanglement is quantified by the random robustness of entanglement ~\cite{Ref1}. We find that corresponding to any ensemble of two equally entangled states, there exists a two-element ensemble of separable states, created by mixing a minimum amount of the maximally mixed state with the entangled states, with the separable pair being harder to distinguish than the entangled one. We refer to this phenomenon as ``More indistinguishability with less entanglement". Notably, a similar phenomenon does not always occur for two-element ensembles of one entangled and one separable state. In the later part of the article,
 we numerically examine the derived bounds, considering ensembles of Haar-uniformly generated two-qubit states of arbitrary but fixed ranks. We find that for two-element ensembles of unequal entanglements, the minimum of the two entanglements must possess a threshold value for the ensemble to exhibit ``More indistinguishability with less entanglement".
 


The remaining parts of the article  is organized as follows. In Sec.~\ref{2}, we present definitions and properties of some well-known quantities that will be needed later in the paper. In Sec.~\ref{iii}, we provide upper and lower bounds on the success probability of discriminating quantum states chosen from any fixed ensemble. Moreover, we present an example where the lower bound can be reached. Considering ensembles of two multipartite-states, utilizing global measurement schemes, and two special situations of random robustness of entanglement of the states, in Sec.~\ref{4}, we find independent lower bounds on the success probability of QSD. Finally, in the same section, we take ensembles consisting of two Haar-uniformly generated two-qubit states and compare the actual optimal probability of successfully discriminating the states of the ensemble with the given upper bound. In the same section we numerically look for the value of minimum entanglement require for the activation of the phenomenon of ``More indistinguishibality with less entanglement" focusing on two-element ensembles of pure states. We provide our concluding remarks in Sec.~\ref{s5}.

\section{Preliminaries}
\label{2}
In this section, we provide a brief description of the requisite quantities, like robustness of $n$-entanglement, random robustness of $n$-entanglement, probability of success in MESD tasks.
\subsection{Robustness of entanglement}

Suppose we have an $m$-partite composite system, $S$. The state of $S$ can be described by a density matrix, $\rho$, which acts on the Hilbert space $\mathcal{H}=\mathcal{H}_1\otimes \mathcal{H}_2\otimes \cdot\cdot\cdot\mathcal{H}_m$, where $\mathcal{H}_i$ describes $i$th part of the entire system, $S$. Let the dimension of $\mathcal{H}$ be $d$. A pure state, $\ket{\psi}\in\mathcal{H}$, is called $k$-separable if it can be written as a product of $k$ pure states, i.e., $\ket{\psi}=\prod_{i=1}^k \ket{\phi_i}$, for any $\{\ket{\phi_i}\}_{i=1}^k$. Clearly $k\leq m$. A density operator acting on $\mathcal{H}$, is referred to as $k$-separable if it can be expressed as a convex mixture of $k$-separable pure states.  For example, if $k=2$, the state is called biseparable, if $k=3$, the state is triseparable, or it is fully separable, if $k=m$. Let us denote the set of all $k$-separable states which acts on $\mathcal{H}$ as $\mathcal M_{k}$. A state which is $k$-separable is also $k'$-separable for all $k'\leq k$.

If an $m$-partite state, $\rho$, can not be written as a convex sum of $k$-separable states, we call it $n$-entangled where $n=m-k+1$. Similar to $k$-separability, the state, $\rho$, which is $n$-entangled must be $n'$ entangled for all $n'\leq n$. For simplicity of notations, from now on, we will always use $n$ to denote $m-k+1$.

 Robustness of $n$-entanglement~\cite{Ref1}, $R_n(\rho||\sigma_{k})$, of  $\rho$ with respect to any $\sigma
_k \in \mathcal{M}_{k}$, is defined as the minimum amount of $\sigma
_k$ needed to be mixed with $\rho$ to make the resulting state $k$-separable. 
Formally, it is defined as
\begin{eqnarray*}
    R_n(\rho||\sigma_{k})=\min r, \\
    \text{such that }\frac{\rho+r\sigma_k}{1+r}\in \mathcal{M}_{k},
\end{eqnarray*}
where $r\geq 0$ and can also be infinite.
Therefore we see, corresponding to every $m$-partite state, $\rho$, there exists a $k$-separable state $\sigma^t$ such that $\sigma^t=\frac{\rho+R_n(\rho||\sigma_{k})\sigma_k}{1+R_n(\rho||\sigma_{k})}$. We name $\sigma^t$ the $k$-separable twin of $\rho$. 

Another interesting quantity, known as the random robustness of $n$-entanglement, is defined as follows:
\begin{eqnarray}
    {R}_n(\rho||\frac{I_d}{d})=\min r \nonumber \\
    \text{such that }\frac{\rho+r I_{d}/d}{1+r}\in \mathcal{M}_{k},\label{eq4}
\end{eqnarray}
where $I_d$ is the identity operator acting on the composite Hilbert space $\mathcal{H}$. That is, random robustness of $n$-entanglement is a special kind of robustness of $n$-entanglement defined for any state with respect to a maximally mixed state. 
Let $\boldsymbol{R}$ be the random robustness of $n$-entanglement of a state $\rho$. This implies $\sigma=\frac{\rho+\boldsymbol{R} I_{d}/d}{1+\boldsymbol{R}}$ is the closest $k$-separable state from the state $\rho$ along the line joining the maximally mixed state and the state $\rho$. We refer such separable states as the \emph{closest} $k$-separable state to $\rho$. 
Note, that $\boldsymbol{R} $ is always a finite quantity since maximally mixed state lies at the interior of the $k$-separable ball ~\cite{SB1,SB2,SB3} of non-zero volume. Therefore, mixing identity by a finite amount is enough to trail $\rho$  into the $k$- separable region. 

Now consider an ensemble, $\eta_m:=\{p_{b},\rho_{b}\}_b$, of $m$-partite states, $\rho_b$, where $p_{b}$ is the probability of finding the state $\rho_{b}$ in the ensemble. Each of the states, $\{\rho_b\}_b$, acts on the same Hilbert space, $\mathcal{H}$, of dimension $d$.
Random robustness of $n$-entanglement of a state, $\rho_b$, of the ensemble, $\eta_m$, is $R_n(\rho_{b}||\frac{I_{d}}{d})$. 
The maximum random robustness of $n$-entanglement of the ensemble, $\eta_m$, is defined as
\begin{equation*}
\mathbb{R}_{n}^M(\eta_m)=\max\limits_{\rho_{b}\in \eta_m}R_n(\rho_{b}||\frac{I_{d}}{d}).
\end{equation*}
Similarly, the minimum random robustness of $n$-entanglement of $\eta_m$ is
\begin{equation*}
\mathbb{R}_{n}^m(\eta_m)=\min\limits_{\rho_{b}\in\eta_m}R_n(\rho_{b}||\frac{I_{d}}{d}).
\end{equation*}

Corresponding to each state $\rho_b$ of the ensemble $\eta_m$ of $m$-partite states, there exists a unique $k$-separable state, $\sigma_b$, such that 
\begin{equation}
  \sigma_b=\frac{\rho_b+R_n(\rho_{b}||\frac{I_{d}}{d}) \frac{I_{d}}{d}}{1+R_n(\rho_{b}||\frac{I_{d}}{d})}.  \nonumber
\end{equation}
 Hence for any ensemble, $\eta_m=\{p_{b},\rho_{b}\}$, we can always define another ensemble, $\epsilon_{\eta_m}=\{p_b,\sigma_b\}$, which would contain the $k$-separable states, $\sigma_b$. We refer to $\epsilon_{\eta_m}$ as the closest ensemble of $k$-separable states to $\eta_m$. 

Since bipartite states consist of only one type of entanglement, i.e., the bipartite entanglement, we will call random robustness of 1-entanglement of bipartite states as simply random robustness of entanglement or RRE. Similarly, the maximum or minimum random robustness of $1$-entanglement of bipartite ensembles, $\eta_2$, will be referred to as the maximum and minimum RRE of $\eta_2$, respectively, and denoted as $\mathbb{R}_1^M(\eta_2)$ and $\mathbb{R}_1^m(\eta_2)$.

\label{secA}

\subsection{Minimum error state discrimination}
Suppose two players, Alice and Bob, play a game. Alice randomly picks an $m$-partite state, say $\rho_b$, from an ensemble, $\eta_m$, and sends the entire state to Bob. Bob knows about the ensemble, $\eta_m$, from which the state has been chosen but does not have any information about the particular chosen state $\rho_b$. Bob's task is to determine which state has been given to Bob. In particular, Bob wants to identify the index, $b$. If he correctly uncovers $b$, he wins. The only way in which Bob can distinguish the states is by performing measurements. For this he may use a positive-operator-valued-measure (POVM), $\{M_{i}\geq 0\}_i$ and $\sum_i M_i = I_d$. If this POVM is measured on the received state, $\rho_b$, the probability that $M_i$ will get clicked is $$q(i|b)=\Tr[\rho_{b}M_{i}].$$
Bob's strategy is when an element, $M_{i}$, of the POVM appears as the measurement outcome Bob guesses the state to be $\rho_i$. Thus for a successful guess, we need the index $i$ to be same as the index $b$ 
The average probability of winning, that is, successfully guessing the correct state, averaged over the entire ensemble, is thus given by,
\begin{equation}
\mathcal{P}=\sum_{b,i}p_{b}\delta_{i,b}\Tr[\rho_{b}M_{i}].\nonumber
\end{equation}
If Bob wants to maximize his chances of winning, he has to find the best POVM that optimizes $\mathcal{P}$. Hence the optimal success probability is given as
\begin{equation}
P_G(\eta_m)=\max\limits_{\{M_{i}\}\in \text{POVM}}\sum_{b,i}p_{b}\delta_{i,b}\Tr[\rho_{b}M_{i}].\label{eq1}
\end{equation}
The game is known as global quantum state discrimination and $P_G(\eta_m)$ is the probability of success.

It is possible to complicate the game a bit further by including another player, say Charu. Alice can again randomly select a state, $\rho_b$, from $\eta_m$, and send some parts of it to Bob and the remaining parts to Charu. Bob and Charu's goal would be to determine exactly which state Alice has sent to them. Depending on the distance between Bob and Charu and their accessibility to communication channels, Bob and Charu can choose a desired scheme for discrimination. For instance, if Bob and Charu both are present in the same location, they can just come together with their respective parts of the system and make global measurements on the whole system, to guess the state. In this case, the maximum probability of guessing can be found using Eq.~\eqref{eq1}, where $\{M_i\}$ are the global POVM performed on the received state, $\rho_b$.  This is again a global quantum state discrimination.

On the other hand, if Bob and Charu live far away and communication expenses are high, the only option left to them is to make local operations on their respective parts without using any classical communication and try to communicate once they already have the outcomes of the measurements to finally make their guess. Then the discrimination scheme is only dependent on local operations (LO). The corresponding optimal probability of success can be written as
\begin{eqnarray}
P_{LO}(\eta_m)=\max\limits_{\{M_{i}\},\{N_j\}\in \text{POVM}}\sum_{i,j,b}p_{b}\delta_{ij,b}\\
\Tr[M_{i}\otimes N_{j}\rho_{b}].\label{eq2}
\end{eqnarray}

If Bob and Charu are wealthy enough to classically communicate in between their measurements, one can cooperatively choose measurements depending on the other's measurement result. This measurement scheme is popularly known as local operations and classical communication (LOCC). In this case, the optimal success probability is
\begin{equation}
P_{LOCC}(\eta_m)=\max\limits_{\{M_{i}\}\in \text{LOCC}}\sum_{i,b}p_{b}\delta_{i,b}\Tr[M_{i}\rho_{b}].\label{eq3}
\end{equation}
    Note, that the the equations ~\ref{eq2},~\ref{eq3} holds true even if multiple parties are playing the game. In a sense, that each of the parties can make local operation (LO) on his/her share of the states to identify it. Or, they can also communicate with each other after making local measurements i.e.(LOCC). In both the scenarios one may define the success probability in a similar fashion as described above, for bipartite case.  
\subsection{Global discrimination between two states}
Let $\eta_m^2=\{p_b,\rho_b\}$ is any arbitrary ensemble of two $m$-partite states, $\rho_1$ and $\rho_2$, i.e., $b$ takes only two different values, 1 and 2. Here superscript 2 implies the ensemble consists of only two distinct states. The optimal probability, $P_G(\eta_m^2)$, of correctly discriminating these two states using global operations is given by the Heltstorm bound, i.e.,~\cite{Ref1} 
\begin{equation}
    P_G\left(\eta_m^2\right)=\frac{1+||p_{1}\rho_1-p_{2}\rho_2||_{1}}{2},\nonumber
\end{equation}
where $||A||_1=\Tr\left(\sqrt{A^\dagger A}\right)$ is the trace norm of the matrix $A$. We would like to mention here some properties of the trace-norm which will be used in the remaining part of the paper:
\begin{enumerate}
    \item 
    \label{p1}
    $||A||_1\geq 0$ for any operator $A$.
    \item 
    \label{p2}
    $||cA||_1 = |c|~||A||_1>0$ for any operator $A$ and complex number~$c$.
    \item 
    \label{p3}
    $||A+B||_1\leq ||A||_1+||B||_1$ for any pair of operators, $A$ and $B$. This property can also be written in the following form
    $||A-B||_1\geq |(||A||-||B||)|$.
\end{enumerate}

\section{Bounds on probability of success}
\label{iii}
In this section, we try to find a relation between the probability of successfully discriminating states of an ensemble and the maximum robustness of $n$-entanglement of the ensemble. From now on we will mostly focus on bipartite ensembles, $\eta_2=\{p_b,\rho_b\}$. 
In the following subsection, we formulate two theorems that bound the success probability of guessing a state, $\rho_b$, chosen from the ensemble $\eta_2$ with RRE of $\rho_b$ and the probability of guessing states that are selected from the closest ensemble of fully separable states, $\epsilon_{\eta_2}$.
\subsection{Upper bound}
\label{3a}
 Let us present an upper bound on the probability of correctly guessing a state picked up from a bipartite ensemble.
\begin{theorem}
The probability, $P_X$, of discriminating the states chosen from any fixed bipartite ensemble $\eta_2:=\{p_{b},\rho_{b}\}$, using any particular discrimination scheme, $X$, is upper bounded as
\begin{equation}
P_X(\eta_2)\leq P_{X}(\epsilon_{\eta_2})\left[1+\mathbb{R}^M_1(\eta_2)\right],
\label{14}
\end{equation}
where  $\epsilon_{\eta_2}$ is the closest ensemble of fully separable states to ${\eta_2}$ and $\mathbb{R}^M_1(\eta_2)$ is the maximum RRE of the ensemble $\eta_2$.
\begin{proof}
Let RRE of the state $\rho_b$ be $\boldsymbol{R}_b=R_1(\rho_b||\frac{I_d}{d})$, where $d$ is the dimension of the Hilbert space, say $\mathcal{H}_1\otimes\mathcal{H}_2$, on which $\rho_b$ acts and $I_d$ is the identity operator that also acts on $\mathcal{H}_1\otimes\mathcal{H}_2$. Hence the separable state, say $\sigma_b$, of the ensemble $\epsilon_{\eta_2}=\{p_b,\sigma_b\}$, is given by
\begin{equation}
\label{eq5}
\sigma_{b}=\frac{\rho_{b}+\boldsymbol{R}_bI_d/{d}}{1+\boldsymbol{R}_b}.
\end{equation}
Since, $I_{d}$ is a positive operator, one can easily show using Eq.~\eqref{eq5} that for every state, $\rho_{b}$, the following inequality holds:
\begin{equation}
\rho_{b}\leq[1+\boldsymbol{R}_b]\sigma_b.\nonumber
\end{equation}
By multiplying both sides with POVM operators, $M_{b}$, and $p_b$, taking trace on both sides,  and then summing over $b$, we get
\begin{equation}
\sum\limits_{b}p_{b}\Tr[M_{b}\rho_{b}]\leq\sum\limits_{b}[1+\boldsymbol{R}_b]p_{b}\Tr[M_{b}\sigma_{b}].
\label{5}
\end{equation}
Employing optimization over a set of POVMs, say $\boldsymbol{X}$, on both sides, we get 
\begin{eqnarray*}
\max\limits_{M_{b}\in\boldsymbol{X}}\sum\limits_{b}p_{b}\Tr[M_{b}\rho_{b}]\leq \max\limits_{M_{b}\in\boldsymbol{X}}\sum\limits_{b}[1+\boldsymbol{R}_b]p_{b}\Tr[M_{b}\sigma_{b}].
\end{eqnarray*}
The above inequality can also be modified as
\begin{eqnarray*}
\max\limits_{M_{b}\in\boldsymbol{X}}\sum\limits_{b}p_{b}\Tr[M_{b}\rho_{b}]\leq \left[1+\mathbb{R}_1^M(\eta_2)\right] \\\max\limits_{M_{b}\in\boldsymbol{X}}\sum\limits_{b}p_{b}\Tr[M_{b}\sigma_{b}],
\end{eqnarray*}
where $\mathbb{R}_1^M(\eta_2)$ is the maximum RRE of $\eta_2$ as defined in Sec.~\ref{secA}. The left-hand side of the above inequality is nothing but the probability, $P_X(\eta_2)$, of winning the state discrimination game with ensemble $\eta_2$, while the term on the right-hand side being multiplied with $\left[1+\mathbb{R}_1^M(\eta_2)\right]$ is the winning probability, $P_X(\epsilon_{\eta_2})$, when the ensemble is $\epsilon_{\eta_2}$. Thus we have
\begin{equation}
    P_X(\eta_2)\leq P_X(\epsilon_{\eta_2})\left[1+\mathbb{R}^M_1(\eta_2)\right].\nonumber
\end{equation}
Here $X$ denotes the measurement scheme which specifies the set, $\boldsymbol{X}$, over which the optimization has been done.
\end{proof}
\label{l1}
\end{theorem}
As we just mentioned, the superscript, $X$, denotes the type of scheme undertaken by the parties to discriminate the states. For example, global measurements, local operations, and local operations and classical communication can be represented by $X=G$, $X=LO$, and $X=LOCC$, respectively. Therefore considering $\boldsymbol{X}$ to be the set of all measurements which can be implemented locally with or without doing any classical communication we get the following Corollary:
\begin{corollary}
The optimal probability, $P_{LO/LOCC}(\eta_2)$, of discriminating bipartite states chosen from an ensemble, $\eta_2$, using local operation without or with classical communication is upper bounded as
$$P_{LO/LOCC}(\eta_2)\leq P_{LO/LOCC}(\epsilon_{\eta_2})\left[1+\mathbb{R}^M_1(\eta_{2})\right],$$
where $P_{LO/LOCC}(\epsilon_{\eta_2})$ is the optimal probability of successfully discriminating the states selected from the closest ensemble $\epsilon_{\eta_2}$ of fully separable state to $\eta_2$ using LO or LOCC and $\mathbb{R}^M_1(\eta_2)$ is the maximum RRE of $\eta_2$.
\label{c1}
\end{corollary}

It is important to note here that none of the Theorem 1 or Corollary 1 depends on the dimension of the bipartite states or the size of the ensemble, $\eta_2$. Though Theorem 1 is stated and proved for bipartite ensembles and the bound depends on RRE of the states, the Theorem can be generalized to any $m$-partite ensemble, $\eta_m=\{p_b,\rho_b\}$. As one can notice from Eq.~\eqref{eq4}, the definition of random robustness of $n$-entanglement involves mixing maximally mixed state with the state, $\rho_b$, of which the entanglement is being quantified to make the resulting state, $\sigma_b$, $k$-separable, where $k=m-n+1$. Therefore corresponding to every $\eta_m$, we can define an ensemble, $\epsilon_{\eta_m}=\{p_b,\sigma_b\}$, of $k$-separable states, that are closest to the states of $\eta_m$. If we generalize the bound on probability of successfully discriminating the states of ensemble, $\eta_m$, for any $m$-partite ensemble, that bound will involve maximum random robustness of $n$-entanglmenet of the ensemble, $\eta_m$, and probability of correctly discriminating the states of the ensemble $\epsilon_{\eta_m}$. Let us state the bound through the following corollary:
\begin{corollary}
The probability, $P_{X}(\eta_m)$, of discriminating the $m-$partite states chosen from an ensemble, $\eta_m$, is upper bounded as
$$P_{X}(\eta_m)\leq P_{X}(\epsilon_{\eta_m})\left[1+\mathbb{R}^M_n(\eta_m)\right],$$
where $\epsilon_{\eta_m}$ is the closest ensemble to $\eta_m$ of $m-n+1$-separable states, ${\sigma_b}$,  and $\mathbb{R}^M_n(\eta_m)$ is the maximum random robustness of $n$-entanglement of the ensemble $\eta_m$.
\label{c2}
\end{corollary}
The above corollary can be easily proved by following the same path of logic as in the proof of Theorem 1.

\textbf{Remark 1.} In deriving the upper bound we have taken the assistance of random robustness of $n$-entanglement. But without loss of generality, one can derive another bound on the success probability of discriminating states drawn from an ensemble, say $\eta_m=\{p_b,\rho_b\}$, considering general robustness of $n$-entanglement of the states, $\rho_b$, with respect to any $k$-separable state, $\sigma_k$, instead of random robustness. Following the exact steps as in Theorem 1, it is easy to see 
    \begin{equation}
       P_X(\eta_m)\leq P_{X}\left(\epsilon_{\eta_m}^{t}\right)\left[1+\mathbb{R'}^M_n(\eta_m)\right].\nonumber
    \end{equation}
    The right-hand side of the above inequality has exactly the same form as of \eqref{14}, with $\epsilon_{\eta_m}$ and $\mathbb{R}^M_n(\eta_m)$ being replaced by $\epsilon_{\eta_m}^{t}=\{p_b,\sigma_b^t\}$ and $\mathbb{R}'^M_n(\eta_m)=\max\limits_{\rho_{b}\in \eta_m}R_n(\rho_{b}||\sigma_k)$, respectively. Here $\sigma^t_b$ is the $k$-separable twin of $\rho_b$. 

\textbf{Remark 2.}
The derived bound suggests that the ratio, 
$P_X(\eta_2)/P_X(\epsilon_{\eta_2})$ can be greater than 1. We will show numerically in the subsequent section that such scenarios do appear 
by generating Haar-uniform ensembles of two two-qubit states. 
Also, for instance, consider an ensemble consisting of two Bell states $\ket{\phi^{+}}$ and $\ket{\psi^{+}}$, with equal probabilities. These two states can be discriminated with certainty using global measurements or even using LOCC~\cite{jon}. Both the states have random robustness of 2. Therefore, the resultant separable ensemble formed by mixing the bell states with identity, has two states that are no longer orthogonal and hence, can not be perfectly discriminated by any operation. This is also an example of non-locality without entanglement~\cite{Ben,NLWE} wherein, a set of separable states is hard to discriminate than the corresponding set  of entangled ones, when any LOCC is allowed for discrimination.

\subsection{Lower bound}
\label{3b}
Now, we attempt to bound the success probability of QSD from below.
\begin{theorem}
\label{lemma1}
    The maximum probability of guessing the states chosen from an ensemble, $\eta_m=\{p_b,\rho_b\}$, of $m$-partite states using any particular measurement scheme, $X$, is lower bounded as
    \begin{equation}\label{eq6}
        P_X(\eta_m)\geq  P_X\left(\epsilon_{\eta_m}\right)\left(1+\mathbb{R}_n^m(\eta_m)\right)-\max_{b}(R_n(\rho_b||\frac{I_d}{d})p_{b}),
    \end{equation}
    where $P_X(\epsilon_{\eta_m})$ is the maximum probability of correctly discriminating states selected from $\epsilon_{\eta_m}=\{p_b,\sigma_b\}$, i.e., the closest $m-n+1$-separable ensemble to $\eta_m$. Here $R_n(\rho_b||\frac{I_d}{d})$ and $\mathbb{R}_n^m(\eta_m)$ are the random robustness of $n$-entanglement of $\rho_b$ and minimum random robustness of $n$-entanglement of $\eta_m$, respectively. 
\end{theorem}
\begin{proof}
    Let the random robustness of $n$-entanglement of the state, $\rho_b$, is $\boldsymbol{R}_b$. Let the closest $m-n+1$-separable state to $\rho_b$, when mixed with maximally mixed state, be$\sigma_b$. This implies, 
    \begin{equation}
        \rho_{b}+\boldsymbol{R}_bI_{d}/d=[1+\boldsymbol{R}_b]\sigma_{b}.\nonumber
    \end{equation}
    Multiplying both sides with $p_b$ and $M_b$ which is an element of a set of POVM outcomes, $\{M_b\}_b$, taking trace, and finally summing over $b$ on both sides we get
    \begin{eqnarray} \sum\limits_{b}p_{b}\Tr\left[M_{b}\left(\rho_{b}+\boldsymbol{R}_bI_{d}/d\right)\right]=
    \sum\limits_{b}p_{b}\Tr\left[\left(1+\boldsymbol{R}_b\right)M_{b}\sigma_{b}\right].\nonumber
    \end{eqnarray}
    This equation implies that
      \begin{eqnarray*}
       \left[\max\limits_{b}(\boldsymbol{R}_bp_{b})\right]\left(\Tr\left[\sum\limits_{b}M_{b}/d\right]\right)+\sum\limits_{b}p_{b}\Tr[M_{b}\rho_{b}]\\ \geq \left[\min\limits_{b}(1+\boldsymbol{R}_b)\right]\sum\limits_{b}p_{b}\Tr[M_{b}\sigma_{b}].        
\end{eqnarray*}
The above inequality is true for any POVM, $\{M_b\}$. Hence it remains valid if we maximize both sides over all POVMs, $\{M_b\}$, which belongs to the specific measurement scheme, $X$. After such maximization, we get
\begin{equation*}
 P_X(\eta_m)\geq (1+\min\limits_{b}\boldsymbol{R}_b)P_X(\epsilon_{\eta_m})-\max_{b}(\boldsymbol{R}_bp_{b}).
\end{equation*}
By putting $\boldsymbol{R}_b=R_n(\rho_b||\frac{I_d}{d})$ and $\min\limits_{b}\boldsymbol{R}_b=\mathbb{R}_n^m(\eta_m)$ in the above inequality we get
\begin{equation}
    P_X(\eta_m)\geq \left(1+\mathbb{R}_n^m(\eta_m)\right)P_X(\epsilon_{\eta_m})-\max_{b}(R_n(\rho_b||I_d)p_{b}). 
   \nonumber
\end{equation}
This completes the proof of Theorem 2.
\end{proof}

We would like to note here that arguments of Corollary~\eqref{c2} and Theorem~\ref{lemma1} remain valid for any measurement scheme, $X$, be it LO, LOCC, global measurements, and are also true for any number of parties involved in the game. 

\subsection{Example}
\label{3c}
There exist various examples where the lower bound provided above, is achievable. Let us discuss one such example. Consider a bipartite ensemble (i.e., $m=2$), $\Bar{\eta}_2^2$, of size 2. The states available in the ensemble are $\rho_{1}=\ket{\phi^{+}}\bra{\phi^{+}}$ and $\rho_{2}=\ket{\phi^{-}}\bra{\phi^{-}} $ with equal probabilities, i.e.,  $p_{1}=p_{2}=\frac{1}{2}$, where $\ket{\phi^{+}}$ and $\ket{\phi^{-}}$ are any two different Bell states. The dimension of the Hilbert space where these Bell states belong is $d=4$. Hence RRE of the Bell states are~\cite{Ref1} $R_1(\rho_{1}||\frac{I_4}{4})=R_1(\rho_{2}||\frac{I_4}{4})=2.$ 
Consider state discrimination using global measurements, i.e., $X=G$. Since the states, $\rho_1$  and $\rho_2$, are orthogonal they can be perfectly discriminated. Thus $P_G\left(\Bar{\eta}_2^2\right)=1$.
The bipartite separable states that are closest to, respectively, $\rho_{1}$ and $\rho_{2}$, with respect to maximally mixed state are
$\sigma_{1}=\frac{1}{3}\left(\ket{\phi^{+}}\bra{\phi^{+}}+\frac{I_4}{2}\right)$ and $\sigma_{2}=\frac{1}{3}\left(\ket{\phi^{-}}\bra{\phi^{-}}+\frac{I_4}{2}\right)$. The probability of optimally discriminating the states of the ensemble $\sigma_{\Bar{\eta}_2^2}=\{p_b,\sigma_b\}_{b=1}^2$ is given by the Helstrom bound~\cite{Hels}, i.e.,  
\begin{equation}
P_G\left(\sigma_{\Bar{\eta}_2^2}\right)=\frac{1}{2}+\frac{1}{2}||p_{1}\sigma_1-p_{2}\sigma_2||_{1}= 
    \frac{2}{3}. \nonumber
\end{equation}
 By putting $m=2$, $X=G$, $P_G\left(\sigma_{\Bar{\eta}_2^2}\right)=2/3$, $P_G\left(\Bar{\eta}_2^2\right)=1$, $R(\rho_{1}||I_{4})=R(\rho_{2}||I_{4})=2$, and $p_1=p_2=\frac{1}{2}$ in the left and right-hand side of the inequality~\eqref{eq6}, it is easy to check that, in this case, the lower bound is exactly achievable.
\begin{figure*}
     \centering
  \includegraphics[scale=0.65]{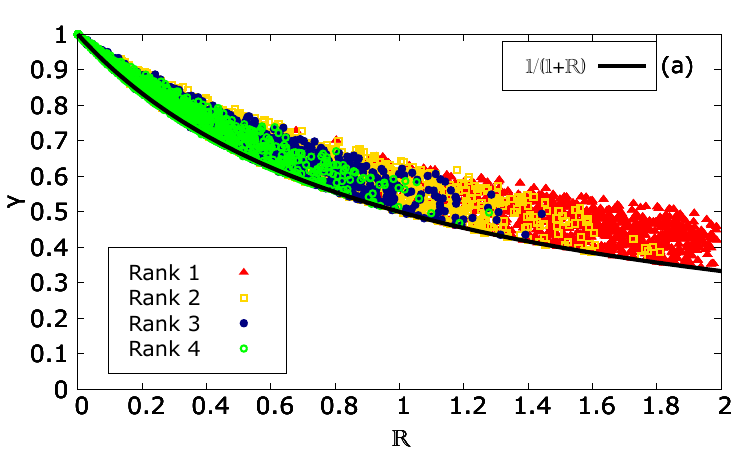}
    \hspace{1.2cm}
\includegraphics[scale=0.65]{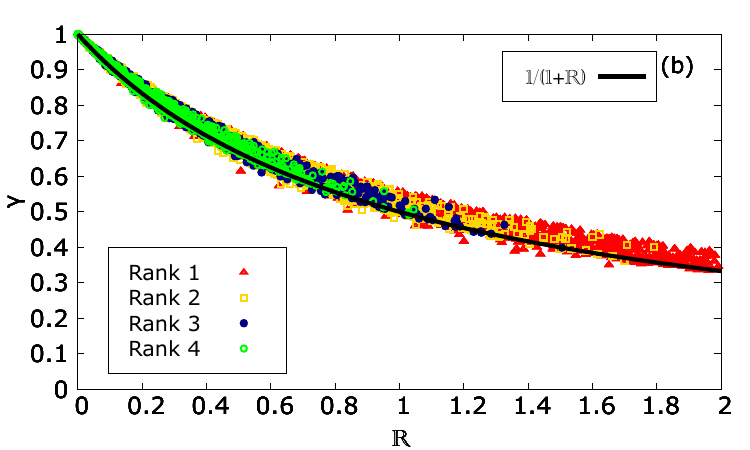}
     \caption{Analysis of the tightness of the upper bound on the probability of successful global discrimination of states. We plot $\gamma(i)$ along the vertical axis with respect to $\mathbb{R}(i)$ presented in the horizontal axis for 
     ensembles, $\eta^2_2$, consisting of Haar-uniformly generated states. In the left panel, the green, blue, yellow, and red points represent ensembles consisting of two rank-4, rank-3, rank-2, and rank-1 states, respectively. The green, blue, yellow, and red points of the right panel represent ensembles having one pure product state and one rank-4, rank-3, rank-2, and rank-1 state, respectively. There are 1000 points for each color in both of the panels, i.e., 1000 ensembles of each type, specified by the rank of the enclosed states. In both of the panels, the black line presents $\gamma=1/({1+\mathbb{R}})$ curve. All the axes are dimensionless. }
     \label{fig1}
 \end{figure*}
\section{The smallest non-trivial ensembles: Ensembles of size two} 
\label{4}
We now focus on ensembles having only two states and global QSD, i.e., $X=G$.
In the first part of this section, we consider two cases and find lower bounds on the probability of correctly distinguishing the pair of states in each of the cases. In the first case, we restrict the two states of the ensemble to have equal random robustness of $n$-entanglement and in the later one, we assume one state of the ensemble to be $k$-separable.
In the last part of the section, we scrutinize these bounds
~on the optimal probability of correctly distinguishing states. 
\subsection{When the two states have equal random robustness of $n$-entanglement}
Let us discuss a game where Alice chooses between the states $\rho_1$ and $\rho_2$ with probability $p_1$ and $p_2$ respectively, from an ensemble of $m$-partite states, $\Bar{\eta}_m^2$. Here again superscript 2 implies the ensemble contains only two distinct states. Based on the nature of the random robustness of entanglement of $\rho_1$ and $\rho_2$, two crucial Theorems can be derived, which we discuss below.

\label{4a}
\begin{theorem}
    The optimal probability $\left(P_{G}\left(\Bar{\eta}_m^2\right)\right)$ of globally discriminating the states of an ensemble $\Bar{\eta}_m^2:=\{p_{b},\rho_{b}\}$, containing two $m$-partite states, $\rho_1$ and $\rho_2$, of equal random robustness of $n$-entanglement is always greater than the optimal probability $\left(P_{G}\left(\epsilon_{\Bar{\eta}_m^2}\right)\right)$ of global discrimination of the states of the ensemble $\epsilon_{\Bar{\eta}_m^2}:=\{p_{b},\sigma_{b}\}$ where $\sigma_b$ is the closest $m-n+1$-separable state to $\rho_b$.
    \label{l3}
\end{theorem}
\begin{proof}
The optimal probability of successfully discriminating two states, $\rho_1$ and $\rho_2$, chosen with probability $p_1$ and $p_2$ using global operations can be found using the Helstrom bound and is given by
 \begin{equation}
     P_{G}\left(\Bar{\eta}_m^2\right)=\frac{1}{2}+\frac{1}{2}||p_{1}\rho_{1}-p_{2}\rho_{2}||_{1}.\label{eq7}
 \end{equation}
Let the random robustness of $n$-entanglement of both of the states be $\boldsymbol{R}$. The $m-n+1$-separable state, $\sigma_b$, closest to $\rho_b$ can be written as
\begin{equation}
\sigma_b=\frac{d\rho_b+\boldsymbol{R}I_d}{d(1+\boldsymbol{R})},\label{eq9}
\end{equation}
where $d$ is the dimension of the Hilbert space on which the state, $\rho_b$, acts and $I_d$ is the identity operator acting on the same Hilbert space.
We can again use the Helstorm bound to find the maximum probability of discriminating states chosen from $\epsilon_{\Bar{\eta}_m^2}$ by global measurement which is given by
\begin{equation} P_{G}({\epsilon}_{\Bar{\eta}_m^2})=\frac{1}{2}+\frac{1}{2}||p_{1}\sigma_{1}-p_{2}\sigma_{2}||_{1}.\label{eq8}
 \end{equation}
 By subtracting Eq.~\eqref{eq8} from Eq.~\eqref{eq7} we get
 \begin{eqnarray}
     P_{G}\left(\Bar{\eta}_m^2\right)-P_{G}\left(\epsilon_{{\Bar{\eta}_m^2}}\right)=\frac{1}{2}(||p_{1}\rho_{1}-p_{2}\rho_{2}||_{1}\nonumber\\-||p_{1}\sigma_{1}-p_{2}\sigma_{2}||_{1}).\label{eq10}
 \end{eqnarray}
Substituting the expression of $\sigma_b$ given to Eq.~\eqref{eq9} in Eq.~\eqref{eq10} and using the mentioned properties of trace norm (see properties~\ref{p1}, \ref{p2}, and \ref{p3}) we get
 \begin{eqnarray}
      P_{G}\left(\Bar{\eta}_m^2\right)-P_{G}\left(\epsilon_{\Bar{\eta}_m^2}\right)\geq \frac{\frac{1}{2}\boldsymbol{R}}{1+\boldsymbol{R}}(||p_{1}\rho_{1}-p_{2}\rho_{2}||_{1}\nonumber\\-|p_{1}-p_{2}|~||I_d/d||_{1}).\label{eq11}
 \end{eqnarray}
Without loss of generality, we can consider $p_1\geq p_2$, i.e., $|p_1-p_2|=p_1-p_2$. Moreover, we know $||I_{d}/d||_{1}=1$. Therefore inequality~\eqref{eq11} reduces to
\begin{eqnarray}
    P_{G}\left(\Bar{\eta}_m^2\right)-P_{G}\left({\epsilon}_{{\Bar{\eta}_m^2}}\right)\geq \frac{\frac{1}{2}\boldsymbol{R}}{1+\boldsymbol{R}}(||p_{1}\rho_{1}-p_{2}\rho_{2}||_{1}-p_{1}+p_{2}).\nonumber
\end{eqnarray}
Again by using the same basic properties of trace norm and the fact that trace norms of normalized density matrices are always equal to unity, we have
\begin{equation}
P_{G}\left(\Bar{\eta}_m^2\right)-P_{G}\left({\epsilon}_{{\Bar{\eta}_m^2}}\right)\geq 0.\nonumber
\label{lb}
\end{equation}
\end{proof}
The above proof does not depend on the dimensions of the states. Since no compact form of the success probability in local state distinguishability of an arbitrary pair of states is available, the above Theorem can not be generalized to local state discrimination.

It is interesting to notice from the above Theorem that corresponding to every pair of equally entangled states, equal in the sense of RRE, there exists a pair of separable states which has lesser global distinguishability compared to the entangled-state pair, We call this phenomenon as " More indistinguihability with less entanglement". We will see in the next subsection that such an order in distinguishability disappears if we move our attention to a pair of states having unequal RRE.
\subsection{When one of the states is separable}
One can notice that Theorem~\ref{l3} is applicable to ensembles that consist of two states of equal RRE. Let us now move to the second case, where we again consider ensembles of two states, but one of the states is restricted to being $k$-separable and the other may have any RRE. 
 
\begin{theorem}
 If an ensemble, $\widetilde{\eta}_m^2:=\{p_{b},\rho_{b}\}$, consisting of two $m$-partite states, $\rho_1$ and $\rho_2$, one of which is $k$-separable, is considered for state discrimination, then the difference between optimal probability $\left(P_{G}\left(\widetilde{\eta}_m^2\right)\right)$ of discriminating the states with global measurements and the optimal probability $\left(P_{G}\left({\epsilon}_{\widetilde{\eta}_m^2}\right)\right)$ of globally discriminating the states chosen from the closest ensemble $\epsilon_{\widetilde{\eta}_m^2}:=\{p_{b},\sigma_{b}\}$ of $k$-separable states to $\widetilde{\eta}_m^2$ is lower bounded by $\frac{\frac{1}{2}\boldsymbol{R}_1}{1+\boldsymbol{R}_1}(||p_{1}\rho_1-p_{2}\rho_2||-1)$.
\label{lemma4}
\end{theorem}
\begin{proof}
    Without loss of generality we can consider the state, $\rho_2$, as separable. Let the random robustness of $n$-entanglement of $\rho_1$ is $\boldsymbol{R}_1$. 
    Following the same way of calculations as done in the proof of Theorem~\ref{l3}, we find, in this case,   
    \begin{eqnarray}
        P_{G}\left(\widetilde{\eta}_m^2\right)-P_{G}\left({\epsilon}_{\widetilde{\eta}_m^2}\right)=\frac{1}{2}||p_{1}\rho_{1}-p_{2}\rho_{2}||_{1}\nonumber\\-\frac{1}{2}\left|\left|\frac{p_{1}}{1+\boldsymbol{R}_1}\left(\rho_{1}+\boldsymbol{R}_1I_{d}/d\right)-p_{2}\rho_{2}\right|\right|_{1}. 
        \label{R0}
    \end{eqnarray}
    Using the properties of trace-norm mentioned in \ref{p1}, \ref{p2}, and \ref{p3},
    and the relations, $||I_{d}||_{1}=d$ and $||\rho_{1}||_{1}=||\rho_{2}||_{1}=1$, one can show that,
    \begin{equation}
         P_{G}\left(\widetilde{\eta}_m^2\right)-P_{G}\left({\epsilon}_{\widetilde{\eta}_m^2}\right)\geq\frac{\frac{1}{2}\boldsymbol{R}_1}{1+\boldsymbol{R}_1}(||p_{1}\rho_1-p_{2}\rho_2||-1).\nonumber
    \end{equation}
\end{proof}

The above proof gives a hint that if one state of an ensemble of size two is separable, then $P_{G}(\widetilde{\eta}_m^2)$ may be less than $P_{G}({\epsilon}_{\widetilde{\eta}_m^2})$. In the next subsection, we will see that there do exist such ensembles, $\widetilde{\eta}_m^2$, for which $P_{G}(\widetilde{\eta}_m^2)<P_{G}({\epsilon}_{\widetilde{\eta}_m^2})$.  
\subsection{Analysis on the tightness of the upper bound}
\label{4b}
In this section, we numerically examine the tightness of the upper bound provided in Theorem \ref{l1}. We consider a global measurement scheme and examine ensembles consisting of a pair of two-qubit states. The ensembles are randomly generated.
The idea is to generate a set of $N$ number of ensembles, $\eta_2^2(i):=\{\{p^i_{b},\rho^i_{b}\}_b\}_i$, each of which consists of two Haar-uniformly selected two-qubit states, $\rho_1^i$ and $\rho_2^i$. 
The probability $p^i_{1}$ is also chosen randomly from (0,1), which automatically specifies the probability, $p^i_2$ as $p^i_2=1-p^i_1$. To estimate the RRE, $R_1(\rho_b^i||\frac{I_4}{4})$, of each of the bipartite states, $\rho_b^i$, of every ensemble, $\eta_2^2(i)$, we have used the numerical non-linear optimization algorithm, NLopt~\cite{ppt}. In particular, through the algorithm, we find for each state, $\rho_b^i$, the minimum value of $\boldsymbol{R}_b^i$, such that the state $\frac{4\rho^i_{b}+\boldsymbol{R}_b^iI_{4}}{4(1+\boldsymbol{R}_b^i)}$ remains positive under partial transposition, which in case of these $2\otimes 2$ states is a necessary and sufficient condition for separability. Here $I_4$ is the identity operator which acts on 4-dimensional Hilbert space. Once we find $\boldsymbol{R}_b^i$ for all $i$ and $b$, we can easily determine maximum RRE, $\mathbb{R}^M_1\left(\eta^2_2(i)\right)$, for all $i$.
The steps, followed in the optimization are 
\begin{enumerate}
    \item Corresponding to each $\rho_{b}$ in the ensemble, a state of the form $\rho_{b}'=\frac{\rho_{b}+\frac{S_{b}I_{4}}{4}}{1+S_{b}}$ is written, The goal is to find the minimum value of $S_{b}$ such that $\rho_{b}'$ is separable.
    \item The partially transposed state corresponding, to each $\rho_{b}'$ is obtained as $\rho_{b}^{T}= \rho_{b}^{T_{(1/2)}}$. Here, $T_{(1/2)}$ represents partial transpose over first or second party. For, $\rho_{b}'$ to be separable the eigenvalues of $\rho_{b}^{T}$, must all be positive.
    \item Next, with the constraint that minimum of the 4 eigenvalues of $\rho_{b}'$ is positive, $S_{b}$ is minimized using the ISRES algorithm, for each index b. Thus, we obtained $S_{b}^{min}$.
    \item As per the definition of robustness, $S_{b}^{min}$ so obtained is nothing but the random robustness of the states $\rho_{b}$ . I.e. $R(\rho_{b}||\frac{I_{4}}{4})=S_{b}^{min}$'.
    \item Now, the robustness of each state is known the next, step is to identify the maximum out of them. This gives, $\mathbb{R}^M_1=\max\limits_{b}\{R(\rho_{b}||\frac{I_{4}}{4})\}$
 \end{enumerate}
Corresponding to each of these randomly generated ensembles, $\eta_2^2(i)$, we can define another ensemble, $\epsilon_{\eta_2^2}(i)=\{p_b^i,\sigma_b^i\}$, where $\sigma_b^i$ is the closest fully separable state to $\rho_b^i$. For every ensemble, $\eta_2^2(i)$ $\left(\epsilon_{\eta_2^2}(i)\right)$, we can find the corresponding probability, $P_G\left(\eta_2^2(i)\right)$ $\left(P_G\left(\epsilon_{\eta_2^2}(i)\right)\right)$, of successfully discriminating the two states, $\rho_1^i$ ($\sigma_1^i$) and $\rho_2^i$ ($\sigma_2^i$), contained in the ensemble employing global operations by making use of the Heltstorm bound. For notational simplicity, we represent $\mathbb{R}^M_1\left(\eta^2_2(i)\right)$, $P_G\left(\eta_2^2(i)\right)$, and $P_G\left(\epsilon_{\eta_2^2}(i)\right)$ by $\mathbb{R}(i)$, $P_\eta^i$, and $P_\epsilon^i$, respectively.

To investigate the bound we define the following quantity:
\begin{equation}
      \gamma(i)= \frac{P_{\eta}^i}{P_\epsilon^i(1+\mathbb{R}(i))}. \nonumber
\end{equation}
The upper bound stated in Theorem~\ref{l1} confirms that $\gamma(i)$ is always less than or equal to 1 for all $i$, where the equality will only be held when the derived bound is achieved. 

 Let us now finally present our numerical findings. Since for Haar-uniform generation of states, we need to specify the rank of the states, we have generated four sets of 1000 ensembles. Each of these sets contains either ensembles of rank 1, rank 2, rank 3, or rank 4 states. So a total of 4000 ensembles are generated. As we mentioned earlier, each of these ensembles consists of two states. In the left panel of Fig.~\ref{fig1}, we present a scattered plot of $\gamma(i)$ with respect to $\mathbb{R}(i)$ for each of the 4000 ensembles. The green, blue, yellow, and red points represent ensembles having rank-1, rank-2, rank-3, and rank-4 states, respectively.

We would like to draw the attention to the following features that can be clearly noticed in the plot:
\begin{itemize}
    \item For small values of $\mathbb{R}$, $\gamma$ is very close to unity, which implies, in such cases, the upper bound is almost achievable. As $\mathbb{R}$ increases, $\gamma$ gets more deviated from unity. This observation suggests that the derived upper bound is tighter for ensembles of weakly entangled states.
    \item In the the same figure, we also plot the curve $1/{(1+\mathbb{R})}$ using a black line. It can be seen from the figure that all of the $(\mathbb{R}(i),\gamma(i))$ points are scattered on or above the black line, stipulating the fact that for each of the generated ensembles, $\eta_2^2(i)$, $P_\eta^i\geq P_\epsilon^i$. Thus we realize that though Theorem~\ref{l3} was derived considering a situation where both of the bipartite states of an ensemble of size two have the same random robustness of entanglement, it also holds true for ensembles consisting of pair of Haar-uniformly generated two-qubit states which of course will always have unequal and non-zero RREs. 
    
\end{itemize}
 
 From the left panel of Fig.~\ref{fig1}, one may conjecture that Theorem~\ref{l3} is true for any ensemble of size two at least for the two-qubit system. To check for the plausibility of this for more general cases, we have generated four more sets of 1000 ensembles where each of them consists of a Haar-uniformly chosen state of fixed rank, say 1, 2, 3, or 4 and one pure product state of the form $\ket{\psi}^i_b\ket{\psi}_b^i$, where $\ket{\psi}_b^i$ is again chosen from Haar-uniform distribution. Hence, in the four generated sets of ensembles, each set consists ensembles of one either rank-1, 2, 3, or 4 state and one product state. Following the same numerical procedure, we again calculate $\gamma(i)$ and $\mathbb{R}(i)$ for all of these 4000 ensembles and plot $\gamma(i)$ with respect to $\mathbb{R}(i)$ in the right panel of Fig.~\ref{fig1}. The following features can be inferred from the plot:
  \begin{itemize}
  \item Unlike the left panel of Fig.~\ref{fig1}, in this plot finite number of points can be seen to be scattered below $1/(1+\mathbb{R})$ curve which has been depicted using a black line in the right panel of Fig.~\ref{fig1}. This points out that if one of the states of some $k^{\text{th}}$ ensemble, say $\eta_2^2(k)$, is separable then there are finite possibilities that the probability, $P_\eta^k$, of discriminating the states of the ensemble using global measurements will be less than $P_{\epsilon}^{k}$, i.e., the probability of correctly discriminating the states of the closest fully separable ensemble, $\epsilon_{\eta_2^2}(k)$, to $\eta_2^2(k)$. 
  \item For ensembles consisting of one pure product state and one entangled state, the upper bound derived in Theorem~\ref{l1} is tighter in the regions where $\mathbb{R}(i)$ is small, and it gets loosened for higher $\mathbb{R}(i)$.
 \end{itemize}


From both the plots of Fig.~\ref{fig1}, we observe that the relation provided in Theorem~\ref{l3} is not only true for ensembles of size two for which the RRE is the same for both the states of the ensemble, but also holds if we generate the two states of the ensemble Haar-uniformly keeping its rank fixed and not restricting entanglement of any of the states in case of two-qubit system. But the Theorem may get violated if we constrain one of the states of the ensemble to be separable, that is when the minimum RRE of the ensemble is kept fixed at zero. This triggers the question that does the inequality, given in Theorem~\ref{l3}, holds for all ensembles, $\eta_2^2$, for which $\mathbb{R}_1^{m}(\eta_2^2)$ is non-zero. If not, what is the cut-off value, $r_c$, such that all ensembles, $\eta_2^2$, which has minimum RRE $\mathbb{R}^m_1(\eta_2^2)\geq r_c$ satisfy the inequality provided in Theorem~\ref{l3}? We try to answer these questions numerically. In particular, we minimize $P_{G}(\eta_2^2)-P_{G}(\epsilon_{\eta_2^2})$, over all two-qubit ensembles, $\eta_2^2:=\{p_b,\rho_b\}$, keeping minimum RRE of $\eta_2^2$ fixed, say at $r$, and denote this optimum value as $\delta P (r)$. During this optimization we have used the exact expression for random robustness which for two-qubit pure state, $(|\Psi\rangle=x_{0}\ket{00}+x_{1}\ket{01}+x_{2}\ket{10}+x_{3}\ket{11})$ is given as,
$R(\ket{\Psi}||\frac{I_4}{4})=4\sqrt{(x_{1}x_{2}-x_{0}x_{3})-{(x_{1}x_{2}-x_{0}x_{3})}^\ast}$. Here $\epsilon_{\eta_2^2}$ is the closest ensemble to $\eta_2^2$ having fully separable states. In short, we define
\begin{eqnarray}
    &&\delta P(r)= \min\limits_{\eta_2^2} \left[P_{G}\left(\eta_2^2\right)- P_{G}\left(\epsilon_{\eta_2^2}\right)\right],\nonumber\\ &&\text{ such that } \mathbb{R}_1^m\left(\eta_2^2\right)=r.\nonumber
\end{eqnarray}
To reduce numerical complexity, we have optimized over ensembles of pure states only.
\begin{figure}
 \centering
  \includegraphics[scale=0.65]{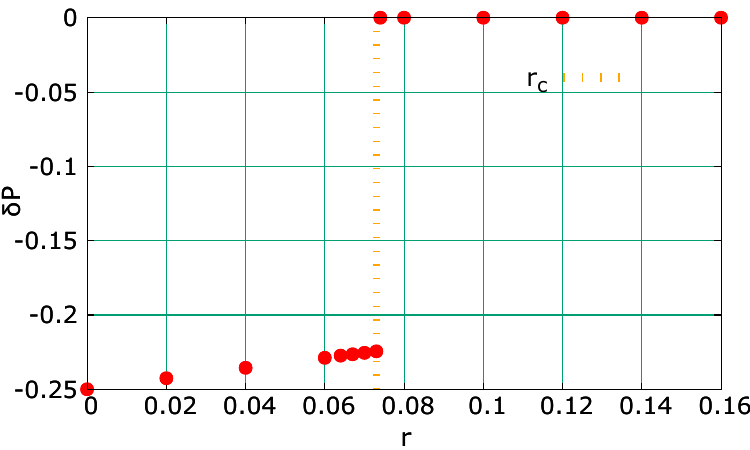}
     \caption{ Behavior of the difference between the probability of globally discriminating states of an ensemble and the states of the closest ensemble of separable states to the former. The plot shows the nature of $\delta P (r)$ (presented along the vertical axis) against fixed minimum RRE, $r$, (presented along the horizontal axis) with red points. The orange dashed line parallel to the vertical axis represents $r=r_c=0.073$ line indicating the value of $r$ above which  $\delta P (r)$ becomes zero. Both the horizontal and vertical axes are dimensionless.}
     \label{fig2}
 \end{figure}
 In Fig.~\ref{fig2}, we plot $\delta P(r)$ for different values of $r$. As it is visible from the plot, for small values of $r$, $\delta P$ is negative which implies that ensembles having small minimum RRE can violate the inequality of the form given in Theorem~\ref{l3}. One can notice further from the figure that for $r\geq 0.073$, $\delta P (r)$ becomes equal to zero up to numerical accuracy. Thus we find $r_c=0.073$. This suffices to conclude that the optimal probability of globally distinguishing the states of any ensembles, $\eta_2^2$, consisting of two two-qubit pure states and having minimum RRE greater than $r_{c}$, is greater or equal to the maximum probability of globally discriminating the states of $\epsilon_{\eta_2^2}$, i.e., the closest ensemble to $\eta_2^2$ containing fully separable states. The figure also indicates that for every $r>r_{c}$ there exists an ensemble, $\Tilde{\eta}_{2,r}^2$, having minimum RRE equal to $r$, such that $P_{G}\left(\Tilde{\eta}_{2,r}^2\right)= P_{G}\left(\epsilon_{\Tilde{\eta}_{2,r}^2}\right)$. This indicates that there exist a threshold value of $r=r_{c}$ above which the phenomenon of ``More indistinguishability with less entanglement" gets activated for two-element ensembles of pure states.
 
\section{Conclusion And Discussion}
\label{s5}
In quantum state discrimination tasks, it is often hard to find the actual probability of success in distinguishing. In this article, considering minimum-error state discrimination, we have provided upper and lower bounds on the optimal probability of successfully guessing the state in terms of the random robustness of entanglement of the states participating in the quantum state discrimination. The bound derived is universal in the sense that it sustains for any measurement scheme, arbitrary dimensions of the systems corresponding to the states to be distinguished, and any number of parties and states involved.

Furthermore, we found that the probability of globally discriminating a pair of multipartite states having equal entanglement is greater or equal to the probability of discriminating the closest separable states to the entangled states. We referred the phenomenon as ``More indistinguishability with less entanglement". 
If we move from the restriction of both states having the same entanglement, and constrain one of the states to be product, then the probability of discriminating the two states may not always exhibit the phenomenon, and for such situations we provide an independent lower bound on the probability of discriminating successfully. Furthermore, we numerically explore how tight the bounds are by examining the global discrimination probability of states selected from Haar-uniformly generated ensembles of two two-qubit states. 
Finally, 
we found that for two-element ensembles of unequal entanglements, the minimum of the two entanglements must possess a threshold value for the ensemble to exhibit the phenomenon of ``More indistinguishability with less entanglement"
\subsection{Acknowledgement}
The research of KS was supported in part by the INFOSYS scholarship. C. S. acknowledges funding from the QuantERA II Programme that has received funding from the European Union’s Horizon 2020 research and innovation programme under Grant Agreement No 101017733 and the F.R.S-FNRS Pint-Multi programme under Grant Agreement R.8014.21,
from the European Union’s Horizon Europe research and innovation programme under the project ``Quantum Security Networks Partnership" (QSNP, grant agreement No 101114043), from the F.R.S-FNRS through the PDR T.0171.22, from the FWO and F.R.S.-FNRS under the Excellence of Science (EOS) programme project 40007526, from the FWO through the BeQuNet SBO project S008323N. Funded by the European Union. Views and opinions expressed are however those of the authors only and do not necessarily reflect those of the European Union, who cannot be held responsible for them.
We acknowledge partial support from the Department of Science and Technology, Government of India through the QuEST grant (grant number 
DST/ICPS/QUST/Theme3/2019/120).
\bibliography{probref}

\begin{thebibliography}{102}%
\makeatletter
\providecommand \@ifxundefined [1]{%
 \@ifx{#1\undefined}
}%
\providecommand \@ifnum [1]{%
 \ifnum #1\expandafter \@firstoftwo
 \else \expandafter \@secondoftwo
 \fi
}%
\providecommand \@ifx [1]{%
 \ifx #1\expandafter \@firstoftwo
 \else \expandafter \@secondoftwo
 \fi
}%
\providecommand \natexlab [1]{#1}%
\providecommand \enquote  [1]{``#1''}%
\providecommand \bibnamefont  [1]{#1}%
\providecommand \bibfnamefont [1]{#1}%
\providecommand \citenamefont [1]{#1}%
\providecommand \href@noop [0]{\@secondoftwo}%
\providecommand \href [0]{\begingroup \@sanitize@url \@href}%
\providecommand \@href[1]{\@@startlink{#1}\@@href}%
\providecommand \@@href[1]{\endgroup#1\@@endlink}%
\providecommand \@sanitize@url [0]{\catcode `\\12\catcode `\$12\catcode `\&12\catcode `\#12\catcode `\^12\catcode `\_12\catcode `\%12\relax}%
\providecommand \@@startlink[1]{}%
\providecommand \@@endlink[0]{}%
\providecommand \url  [0]{\begingroup\@sanitize@url \@url }%
\providecommand \@url [1]{\endgroup\@href {#1}{\urlprefix }}%
\providecommand \urlprefix  [0]{URL }%
\providecommand \Eprint [0]{\href }%
\providecommand \doibase [0]{http://dx.doi.org/}%
\providecommand \selectlanguage [0]{\@gobble}%
\providecommand \bibinfo  [0]{\@secondoftwo}%
\providecommand \bibfield  [0]{\@secondoftwo}%
\providecommand \translation [1]{[#1]}%
\providecommand \BibitemOpen [0]{}%
\providecommand \bibitemStop [0]{}%
\providecommand \bibitemNoStop [0]{.\EOS\space}%
\providecommand \EOS [0]{\spacefactor3000\relax}%
\providecommand \BibitemShut  [1]{\csname bibitem#1\endcsname}%
\let\auto@bib@innerbib\@empty
\bibitem [{\citenamefont {Sacchi}(2005)}]{CD1}%
  \BibitemOpen
  \bibfield  {author} {\bibinfo {author} {\bibfnamefont {M.~F.}\ \bibnamefont {Sacchi}},\ }\bibfield  {title} {\enquote {\bibinfo {title} {Minimum error discrimination of pauli channels},}\ }\href {https://api.semanticscholar.org/CorpusID:11544775} {\bibfield  {journal} {\bibinfo  {journal} {J. opt., B Quantum semiclass.}\ }\textbf {\bibinfo {volume} {7}},\ \bibinfo {pages} {S333} (\bibinfo {year} {2005})}\BibitemShut {NoStop}%
\bibitem [{\citenamefont {Lei}\ \emph {et~al.}(2023)\citenamefont {Lei}, \citenamefont {Cao}, \citenamefont {Kumar},\ and\ \citenamefont {Wu}}]{CD2}%
  \BibitemOpen
  \bibfield  {author} {\bibinfo {author} {\bibfnamefont {Q.}~\bibnamefont {Lei}}, \bibinfo {author} {\bibfnamefont {L.}~\bibnamefont {Cao}}, \bibinfo {author} {\bibfnamefont {A.}~\bibnamefont {Kumar}}, \ and\ \bibinfo {author} {\bibfnamefont {J.}~\bibnamefont {Wu}},\ }\bibfield  {title} {\enquote {\bibinfo {title} {Dilation, discrimination and uhlmann’s theorem of link products of quantum channels},}\ }\href {https://doi.org/10.48550/arXiv.2309.03052} {\bibfield  {journal} {\bibinfo  {journal} {Chin. Phys. B}\ } (\bibinfo {year} {2023})}\BibitemShut {NoStop}%
\bibitem [{\citenamefont {Ambainis}\ \emph {et~al.}(2008)\citenamefont {Ambainis}, \citenamefont {Leung}, \citenamefont {Mancinska},\ and\ \citenamefont {Ozols}}]{Rac1}%
  \BibitemOpen
  \bibfield  {author} {\bibinfo {author} {\bibfnamefont {A.}~\bibnamefont {Ambainis}}, \bibinfo {author} {\bibfnamefont {D.}~\bibnamefont {Leung}}, \bibinfo {author} {\bibfnamefont {L.}~\bibnamefont {Mancinska}}, \ and\ \bibinfo {author} {\bibfnamefont {M.}~\bibnamefont {Ozols}},\ }\bibfield  {title} {\enquote {\bibinfo {title} {Quantum random access codes with shared randomness},}\ }\href {https://doi.org/10.48550/arXiv.0810.2937} {\bibfield  {journal} {\bibinfo  {journal} {arXiv preprint arXiv:0810.2937}\ } (\bibinfo {year} {2008})}\BibitemShut {NoStop}%
\bibitem [{\citenamefont {Roy}\ and\ \citenamefont {Pan}(2023)}]{Rac2}%
  \BibitemOpen
  \bibfield  {author} {\bibinfo {author} {\bibfnamefont {P.}~\bibnamefont {Roy}}\ and\ \bibinfo {author} {\bibfnamefont {A.}~\bibnamefont {Pan}},\ }\bibfield  {title} {\enquote {\bibinfo {title} {Generalized parity-oblivious communication games powered by quantum preparation contextuality},}\ }\href {https://doi.org/10.48550/arXiv.2311.04490} {\bibfield  {journal} {\bibinfo  {journal} {arXiv preprint arXiv:2311.04490}\ } (\bibinfo {year} {2023})}\BibitemShut {NoStop}%
\bibitem [{\citenamefont {Bennett}\ and\ \citenamefont {Wiesner}(1992)}]{dc1}%
  \BibitemOpen
  \bibfield  {author} {\bibinfo {author} {\bibfnamefont {C.~H.}\ \bibnamefont {Bennett}}\ and\ \bibinfo {author} {\bibfnamefont {S.~J.}\ \bibnamefont {Wiesner}},\ }\bibfield  {title} {\enquote {\bibinfo {title} {Communication via one- and two-particle operators on einstein-podolsky-rosen states},}\ }\href {https://link.aps.org/doi/10.1103/PhysRevLett.69.2881} {\bibfield  {journal} {\bibinfo  {journal} {Phys. Rev. Lett.}\ }\textbf {\bibinfo {volume} {69}},\ \bibinfo {pages} {2881} (\bibinfo {year} {1992})}\BibitemShut {NoStop}%
\bibitem [{\citenamefont {Hiroshima}(2001)}]{dc2}%
  \BibitemOpen
  \bibfield  {author} {\bibinfo {author} {\bibfnamefont {T.}~\bibnamefont {Hiroshima}},\ }\bibfield  {title} {\enquote {\bibinfo {title} {Optimal dense coding with mixed state entanglement},}\ }\href {https://doi.org/10.1088/0305-4470/34/35/316} {\bibfield  {journal} {\bibinfo  {journal} {J. Phys. A: Math. Theor.}\ }\textbf {\bibinfo {volume} {34}},\ \bibinfo {pages} {6907} (\bibinfo {year} {2001})}\BibitemShut {NoStop}%
\bibitem [{\citenamefont {Liu}\ \emph {et~al.}(2002)\citenamefont {Liu}, \citenamefont {Long}, \citenamefont {Tong},\ and\ \citenamefont {Li}}]{d3}%
  \BibitemOpen
  \bibfield  {author} {\bibinfo {author} {\bibfnamefont {X.~S.}\ \bibnamefont {Liu}}, \bibinfo {author} {\bibfnamefont {G.~L.}\ \bibnamefont {Long}}, \bibinfo {author} {\bibfnamefont {D.~M.}\ \bibnamefont {Tong}}, \ and\ \bibinfo {author} {\bibfnamefont {Feng}\ \bibnamefont {Li}},\ }\bibfield  {title} {\enquote {\bibinfo {title} {General scheme for superdense coding between multiparties},}\ }\href {https://link.aps.org/doi/10.1103/PhysRevA.65.022304} {\bibfield  {journal} {\bibinfo  {journal} {Phys. Rev. A}\ }\textbf {\bibinfo {volume} {65}},\ \bibinfo {pages} {022304} (\bibinfo {year} {2002})}\BibitemShut {NoStop}%
\bibitem [{\citenamefont {Plesch}\ and\ \citenamefont {Bu\ifmmode~\check{z}\else \v{z}\fi{}ek}(2003)}]{d4}%
  \BibitemOpen
  \bibfield  {author} {\bibinfo {author} {\bibfnamefont {M.}~\bibnamefont {Plesch}}\ and\ \bibinfo {author} {\bibfnamefont {V.}~\bibnamefont {Bu\ifmmode~\check{z}\else \v{z}\fi{}ek}},\ }\bibfield  {title} {\enquote {\bibinfo {title} {Entangled graphs: Bipartite entanglement in multiqubit systems},}\ }\href {https://link.aps.org/doi/10.1103/PhysRevA.67.012322} {\bibfield  {journal} {\bibinfo  {journal} {Phys. Rev. A}\ }\textbf {\bibinfo {volume} {67}},\ \bibinfo {pages} {012322} (\bibinfo {year} {2003})}\BibitemShut {NoStop}%
\bibitem [{\citenamefont {Bru\ss{}}\ \emph {et~al.}(2004)\citenamefont {Bru\ss{}}, \citenamefont {D'Ariano}, \citenamefont {Lewenstein}, \citenamefont {Macchiavello}, \citenamefont {Sen(De)},\ and\ \citenamefont {Sen}}]{dc5}%
  \BibitemOpen
  \bibfield  {author} {\bibinfo {author} {\bibfnamefont {D.}~\bibnamefont {Bru\ss{}}}, \bibinfo {author} {\bibfnamefont {G.~M.}\ \bibnamefont {D'Ariano}}, \bibinfo {author} {\bibfnamefont {M.}~\bibnamefont {Lewenstein}}, \bibinfo {author} {\bibfnamefont {C.}~\bibnamefont {Macchiavello}}, \bibinfo {author} {\bibfnamefont {A.}~\bibnamefont {Sen(De)}}, \ and\ \bibinfo {author} {\bibfnamefont {U.}~\bibnamefont {Sen}},\ }\bibfield  {title} {\enquote {\bibinfo {title} {Distributed quantum dense coding},}\ }\href {https://link.aps.org/doi/10.1103/PhysRevLett.93.210501} {\bibfield  {journal} {\bibinfo  {journal} {Phys. Rev. Lett.}\ }\textbf {\bibinfo {volume} {93}},\ \bibinfo {pages} {210501} (\bibinfo {year} {2004})}\BibitemShut {NoStop}%
\bibitem [{\citenamefont {Horodecki}\ and\ \citenamefont {Piani}(2012)}]{dc6}%
  \BibitemOpen
  \bibfield  {author} {\bibinfo {author} {\bibfnamefont {M.}~\bibnamefont {Horodecki}}\ and\ \bibinfo {author} {\bibfnamefont {M.}~\bibnamefont {Piani}},\ }\bibfield  {title} {\enquote {\bibinfo {title} {On quantum advantage in dense coding},}\ }\href {https://doi.org/10.1088/1751-8113/45/10/105306} {\bibfield  {journal} {\bibinfo  {journal} {J. Phys. A: Math. Theor.}\ }\textbf {\bibinfo {volume} {45}},\ \bibinfo {pages} {105306} (\bibinfo {year} {2012})}\BibitemShut {NoStop}%
\bibitem [{\citenamefont {Paris}(2009)}]{MR1}%
  \BibitemOpen
  \bibfield  {author} {\bibinfo {author} {\bibfnamefont {M.~G.}\ \bibnamefont {Paris}},\ }\bibfield  {title} {\enquote {\bibinfo {title} {Quantum estimation for quantum technology},}\ }\href {https://doi.org/10.48550/arXiv.0804.2981} {\bibfield  {journal} {\bibinfo  {journal} {Int. J. Quantum Inf}\ }\textbf {\bibinfo {volume} {7}},\ \bibinfo {pages} {125} (\bibinfo {year} {2009})}\BibitemShut {NoStop}%
\bibitem [{\citenamefont {Escher}\ \emph {et~al.}(2011)\citenamefont {Escher}, \citenamefont {de~Matos~Filho},\ and\ \citenamefont {Davidovich}}]{MR0}%
  \BibitemOpen
  \bibfield  {author} {\bibinfo {author} {\bibfnamefont {BM.}\ \bibnamefont {Escher}}, \bibinfo {author} {\bibfnamefont {R.~L.}\ \bibnamefont {de~Matos~Filho}}, \ and\ \bibinfo {author} {\bibfnamefont {L.}~\bibnamefont {Davidovich}},\ }\bibfield  {title} {\enquote {\bibinfo {title} {General framework for estimating the ultimate precision limit in noisy quantum-enhanced metrology},}\ }\href {https://doi.org/10.1038/nphys1958} {\bibfield  {journal} {\bibinfo  {journal} {Nat. Phys}\ }\textbf {\bibinfo {volume} {7}},\ \bibinfo {pages} {406} (\bibinfo {year} {2011})}\BibitemShut {NoStop}%
\bibitem [{\citenamefont {Giovannetti}\ \emph {et~al.}(2011)\citenamefont {Giovannetti}, \citenamefont {Lloyd},\ and\ \citenamefont {Maccone}}]{MR2}%
  \BibitemOpen
  \bibfield  {author} {\bibinfo {author} {\bibfnamefont {V.}~\bibnamefont {Giovannetti}}, \bibinfo {author} {\bibfnamefont {S.}~\bibnamefont {Lloyd}}, \ and\ \bibinfo {author} {\bibfnamefont {L.}~\bibnamefont {Maccone}},\ }\bibfield  {title} {\enquote {\bibinfo {title} {Advances in quantum metrology},}\ }\href {https://doi.org/10.1038/nphoton.2011.35} {\bibfield  {journal} {\bibinfo  {journal} {Nat. Photonics}\ }\textbf {\bibinfo {volume} {5}},\ \bibinfo {pages} {222} (\bibinfo {year} {2011})}\BibitemShut {NoStop}%
\bibitem [{\citenamefont {T{\'o}th}\ and\ \citenamefont {Apellaniz}(2014)}]{MR3}%
  \BibitemOpen
  \bibfield  {author} {\bibinfo {author} {\bibfnamefont {G.}~\bibnamefont {T{\'o}th}}\ and\ \bibinfo {author} {\bibfnamefont {I.}~\bibnamefont {Apellaniz}},\ }\bibfield  {title} {\enquote {\bibinfo {title} {Quantum metrology from a quantum information science perspective},}\ }\href {https://doi.org/10.48550/arXiv.0804.2981} {\bibfield  {journal} {\bibinfo  {journal} {J. Phys. A Math. Theor. .}\ }\textbf {\bibinfo {volume} {47}},\ \bibinfo {pages} {424006} (\bibinfo {year} {2014})}\BibitemShut {NoStop}%
\bibitem [{\citenamefont {Pezz\`e}\ \emph {et~al.}(2018)\citenamefont {Pezz\`e}, \citenamefont {Smerzi}, \citenamefont {Oberthaler}, \citenamefont {Schmied},\ and\ \citenamefont {Treutlein}}]{MR4}%
  \BibitemOpen
  \bibfield  {author} {\bibinfo {author} {\bibfnamefont {L.}~\bibnamefont {Pezz\`e}}, \bibinfo {author} {\bibfnamefont {A.}~\bibnamefont {Smerzi}}, \bibinfo {author} {\bibfnamefont {M.~K.}\ \bibnamefont {Oberthaler}}, \bibinfo {author} {\bibfnamefont {R.}~\bibnamefont {Schmied}}, \ and\ \bibinfo {author} {\bibfnamefont {P.}~\bibnamefont {Treutlein}},\ }\bibfield  {title} {\enquote {\bibinfo {title} {Quantum metrology with nonclassical states of atomic ensembles},}\ }\href {https://link.aps.org/doi/10.1103/RevModPhys.90.035005} {\bibfield  {journal} {\bibinfo  {journal} {Rev. Mod. Phys.}\ }\textbf {\bibinfo {volume} {90}},\ \bibinfo {pages} {035005} (\bibinfo {year} {2018})}\BibitemShut {NoStop}%
\bibitem [{\citenamefont {Braun}\ \emph {et~al.}(2018)\citenamefont {Braun}, \citenamefont {Adesso}, \citenamefont {Benatti}, \citenamefont {Floreanini}, \citenamefont {Marzolino}, \citenamefont {Mitchell},\ and\ \citenamefont {Pirandola}}]{MR5}%
  \BibitemOpen
  \bibfield  {author} {\bibinfo {author} {\bibfnamefont {D.}~\bibnamefont {Braun}}, \bibinfo {author} {\bibfnamefont {G.}~\bibnamefont {Adesso}}, \bibinfo {author} {\bibfnamefont {F.}~\bibnamefont {Benatti}}, \bibinfo {author} {\bibfnamefont {R.}~\bibnamefont {Floreanini}}, \bibinfo {author} {\bibfnamefont {U.}~\bibnamefont {Marzolino}}, \bibinfo {author} {\bibfnamefont {M.~W.}\ \bibnamefont {Mitchell}}, \ and\ \bibinfo {author} {\bibfnamefont {S.}~\bibnamefont {Pirandola}},\ }\bibfield  {title} {\enquote {\bibinfo {title} {Quantum-enhanced measurements without entanglement},}\ }\href {https://link.aps.org/doi/10.1103/RevModPhys.90.035006} {\bibfield  {journal} {\bibinfo  {journal} {Rev. Mod. Phys.}\ }\textbf {\bibinfo {volume} {90}},\ \bibinfo {pages} {035006} (\bibinfo {year} {2018})}\BibitemShut {NoStop}%
\bibitem [{\citenamefont {Wootters}\ and\ \citenamefont {Zurek}(1982)}]{I1}%
  \BibitemOpen
  \bibfield  {author} {\bibinfo {author} {\bibfnamefont {W.~K.}\ \bibnamefont {Wootters}}\ and\ \bibinfo {author} {\bibfnamefont {W.~H.}\ \bibnamefont {Zurek}},\ }\bibfield  {title} {\enquote {\bibinfo {title} {A single quantum cannot be cloned},}\ }\href {https://api.semanticscholar.org/CorpusID:4339227} {\bibfield  {journal} {\bibinfo  {journal} {Nature}\ }\textbf {\bibinfo {volume} {299}},\ \bibinfo {pages} {803} (\bibinfo {year} {1982})}\BibitemShut {NoStop}%
\bibitem [{\citenamefont {Yuen}\ \emph {et~al.}(1975)\citenamefont {Yuen}, \citenamefont {Kennedy},\ and\ \citenamefont {Lax}}]{Me0}%
  \BibitemOpen
  \bibfield  {author} {\bibinfo {author} {\bibfnamefont {H.}~\bibnamefont {Yuen}}, \bibinfo {author} {\bibfnamefont {R.}~\bibnamefont {Kennedy}}, \ and\ \bibinfo {author} {\bibfnamefont {M.}~\bibnamefont {Lax}},\ }\bibfield  {title} {\enquote {\bibinfo {title} {Optimum testing of multiple hypotheses in quantum detection theory},}\ }\href {https://ieeexplore.ieee.org/document/1055351} {\bibfield  {journal} {\bibinfo  {journal} {IEEE Trans. Inf. Theory}\ }\textbf {\bibinfo {volume} {21}},\ \bibinfo {pages} {125} (\bibinfo {year} {1975})}\BibitemShut {NoStop}%
\bibitem [{\citenamefont {Helstrom}(1969)}]{Hels}%
  \BibitemOpen
  \bibfield  {author} {\bibinfo {author} {\bibfnamefont {C.~W.}\ \bibnamefont {Helstrom}},\ }\bibfield  {title} {\enquote {\bibinfo {title} {Quantum detection and estimation theory},}\ }\href {https://api.semanticscholar.org/CorpusID:121571330} {\bibfield  {journal} {\bibinfo  {journal} {J Stat Phys}\ }\textbf {\bibinfo {volume} {1}},\ \bibinfo {pages} {231} (\bibinfo {year} {1969})}\BibitemShut {NoStop}%
\bibitem [{\citenamefont {Herzog}\ and\ \citenamefont {Bergou}(2002)}]{Me1}%
  \BibitemOpen
  \bibfield  {author} {\bibinfo {author} {\bibfnamefont {U.}~\bibnamefont {Herzog}}\ and\ \bibinfo {author} {\bibfnamefont {J.~A.}\ \bibnamefont {Bergou}},\ }\bibfield  {title} {\enquote {\bibinfo {title} {Minimum-error discrimination between subsets of linearly dependent quantum states},}\ }\href {https://link.aps.org/doi/10.1103/PhysRevA.65.050305} {\bibfield  {journal} {\bibinfo  {journal} {Phys. Rev. A}\ }\textbf {\bibinfo {volume} {65}},\ \bibinfo {pages} {050305} (\bibinfo {year} {2002})}\BibitemShut {NoStop}%
\bibitem [{\citenamefont {Je\ifmmode~\check{z}\else \v{z}\fi{}ek}\ \emph {et~al.}(2002)\citenamefont {Je\ifmmode~\check{z}\else \v{z}\fi{}ek}, \citenamefont {\ifmmode \check{R}\else \v{R}\fi{}eh\'a\ifmmode~\check{c}\else \v{c}\fi{}ek},\ and\ \citenamefont {Fiur\'a\ifmmode~\check{s}\else \v{s}\fi{}ek}}]{Me2}%
  \BibitemOpen
  \bibfield  {author} {\bibinfo {author} {\bibfnamefont {M.}~\bibnamefont {Je\ifmmode~\check{z}\else \v{z}\fi{}ek}}, \bibinfo {author} {\bibfnamefont {J.}~\bibnamefont {\ifmmode \check{R}\else \v{R}\fi{}eh\'a\ifmmode~\check{c}\else \v{c}\fi{}ek}}, \ and\ \bibinfo {author} {\bibfnamefont {J.}~\bibnamefont {Fiur\'a\ifmmode~\check{s}\else \v{s}\fi{}ek}},\ }\bibfield  {title} {\enquote {\bibinfo {title} {Finding optimal strategies for minimum-error quantum-state discrimination},}\ }\href {https://link.aps.org/doi/10.1103/PhysRevA.65.060301} {\bibfield  {journal} {\bibinfo  {journal} {Phys. Rev. A}\ }\textbf {\bibinfo {volume} {65}},\ \bibinfo {pages} {060301} (\bibinfo {year} {2002})}\BibitemShut {NoStop}%
\bibitem [{\citenamefont {Chou}\ and\ \citenamefont {Hsu}(2003)}]{Me3}%
  \BibitemOpen
  \bibfield  {author} {\bibinfo {author} {\bibfnamefont {C.-L.}\ \bibnamefont {Chou}}\ and\ \bibinfo {author} {\bibfnamefont {L.~Y.}\ \bibnamefont {Hsu}},\ }\bibfield  {title} {\enquote {\bibinfo {title} {Minimum-error discrimination between symmetric mixed quantum states},}\ }\href {https://link.aps.org/doi/10.1103/PhysRevA.68.042305} {\bibfield  {journal} {\bibinfo  {journal} {Phys. Rev. A}\ }\textbf {\bibinfo {volume} {68}},\ \bibinfo {pages} {042305} (\bibinfo {year} {2003})}\BibitemShut {NoStop}%
\bibitem [{\citenamefont {Herzog}\ and\ \citenamefont {Bergou}(2004)}]{Me4}%
  \BibitemOpen
  \bibfield  {author} {\bibinfo {author} {\bibfnamefont {U.}~\bibnamefont {Herzog}}\ and\ \bibinfo {author} {\bibfnamefont {J.~A.}\ \bibnamefont {Bergou}},\ }\bibfield  {title} {\enquote {\bibinfo {title} {Distinguishing mixed quantum states: Minimum-error discrimination versus optimum unambiguous discrimination},}\ }\href {https://link.aps.org/doi/10.1103/PhysRevA.70.022302} {\bibfield  {journal} {\bibinfo  {journal} {Phys. Rev. A}\ }\textbf {\bibinfo {volume} {70}},\ \bibinfo {pages} {022302} (\bibinfo {year} {2004})}\BibitemShut {NoStop}%
\bibitem [{\citenamefont {Qiu}(2008{\natexlab{a}})}]{Me5}%
  \BibitemOpen
  \bibfield  {author} {\bibinfo {author} {\bibfnamefont {D.}~\bibnamefont {Qiu}},\ }\bibfield  {title} {\enquote {\bibinfo {title} {Minimum-error discrimination between mixed quantum states},}\ }\href {\doibase 10.1103/PhysRevA.77.012328} {\bibfield  {journal} {\bibinfo  {journal} {Phys. Rev. A}\ }\textbf {\bibinfo {volume} {77}},\ \bibinfo {pages} {012328} (\bibinfo {year} {2008}{\natexlab{a}})}\BibitemShut {NoStop}%
\bibitem [{\citenamefont {Barnett}\ and\ \citenamefont {Croke}(2008)}]{Me6}%
  \BibitemOpen
  \bibfield  {author} {\bibinfo {author} {\bibfnamefont {S.~M.}\ \bibnamefont {Barnett}}\ and\ \bibinfo {author} {\bibfnamefont {S.}~\bibnamefont {Croke}},\ }\bibfield  {title} {\enquote {\bibinfo {title} {On the conditions for discrimination between quantum states with minimum error},}\ }\href {https://doi.org/10.1088/1751-8113/42/6/062001 Focus to learn more} {\bibfield  {journal} {\bibinfo  {journal} {arXiv preprint arXiv:0810.1919}\ } (\bibinfo {year} {2008})}\BibitemShut {NoStop}%
\bibitem [{\citenamefont {Qiu}\ and\ \citenamefont {Li}(2010{\natexlab{a}})}]{Me7}%
  \BibitemOpen
  \bibfield  {author} {\bibinfo {author} {\bibfnamefont {D.}~\bibnamefont {Qiu}}\ and\ \bibinfo {author} {\bibfnamefont {L.}~\bibnamefont {Li}},\ }\bibfield  {title} {\enquote {\bibinfo {title} {Minimum-error discrimination of quantum states: Bounds and comparisons},}\ }\href {\doibase 10.1103/PhysRevA.81.042329} {\bibfield  {journal} {\bibinfo  {journal} {Phys. Rev. A}\ }\textbf {\bibinfo {volume} {81}},\ \bibinfo {pages} {042329} (\bibinfo {year} {2010}{\natexlab{a}})}\BibitemShut {NoStop}%
\bibitem [{\citenamefont {Assalini}\ \emph {et~al.}(2010)\citenamefont {Assalini}, \citenamefont {Cariolaro},\ and\ \citenamefont {Pierobon}}]{Me8}%
  \BibitemOpen
  \bibfield  {author} {\bibinfo {author} {\bibfnamefont {A.}~\bibnamefont {Assalini}}, \bibinfo {author} {\bibfnamefont {G.}~\bibnamefont {Cariolaro}}, \ and\ \bibinfo {author} {\bibfnamefont {G.}~\bibnamefont {Pierobon}},\ }\bibfield  {title} {\enquote {\bibinfo {title} {Efficient optimal minimum error discrimination of symmetric quantum states},}\ }\href {https://link.aps.org/doi/10.1103/PhysRevA.81.012315} {\bibfield  {journal} {\bibinfo  {journal} {Phys. Rev. A}\ }\textbf {\bibinfo {volume} {81}},\ \bibinfo {pages} {012315} (\bibinfo {year} {2010})}\BibitemShut {NoStop}%
\bibitem [{\citenamefont {Lu}\ \emph {et~al.}(2010{\natexlab{a}})\citenamefont {Lu}, \citenamefont {Coish}, \citenamefont {Kaltenbaek}, \citenamefont {Hamel}, \citenamefont {Croke},\ and\ \citenamefont {Resch}}]{Me10}%
  \BibitemOpen
  \bibfield  {author} {\bibinfo {author} {\bibfnamefont {Y.}~\bibnamefont {Lu}}, \bibinfo {author} {\bibfnamefont {N.}~\bibnamefont {Coish}}, \bibinfo {author} {\bibfnamefont {R.}~\bibnamefont {Kaltenbaek}}, \bibinfo {author} {\bibfnamefont {D.~R.}\ \bibnamefont {Hamel}}, \bibinfo {author} {\bibfnamefont {S.}~\bibnamefont {Croke}}, \ and\ \bibinfo {author} {\bibfnamefont {K.~J.}\ \bibnamefont {Resch}},\ }\bibfield  {title} {\enquote {\bibinfo {title} {Minimum-error discrimination of entangled quantum states},}\ }\href {https://link.aps.org/doi/10.1103/PhysRevA.82.042340} {\bibfield  {journal} {\bibinfo  {journal} {Phys. Rev. A}\ }\textbf {\bibinfo {volume} {82}},\ \bibinfo {pages} {042340} (\bibinfo {year} {2010}{\natexlab{a}})}\BibitemShut {NoStop}%
\bibitem [{\citenamefont {Jafarizadeh}\ \emph {et~al.}(2011{\natexlab{a}})\citenamefont {Jafarizadeh}, \citenamefont {Sufiani},\ and\ \citenamefont {Mazhari~K.}}]{Me9}%
  \BibitemOpen
  \bibfield  {author} {\bibinfo {author} {\bibfnamefont {M.~A.}\ \bibnamefont {Jafarizadeh}}, \bibinfo {author} {\bibfnamefont {R.}~\bibnamefont {Sufiani}}, \ and\ \bibinfo {author} {\bibfnamefont {Y.}~\bibnamefont {Mazhari~K.}},\ }\bibfield  {title} {\enquote {\bibinfo {title} {Minimum error discrimination between similarity-transformed quantum states},}\ }\href {https://link.aps.org/doi/10.1103/PhysRevA.84.012102} {\bibfield  {journal} {\bibinfo  {journal} {Phys. Rev. A}\ }\textbf {\bibinfo {volume} {84}},\ \bibinfo {pages} {012102} (\bibinfo {year} {2011}{\natexlab{a}})}\BibitemShut {NoStop}%
\bibitem [{\citenamefont {Jafarizadeh}\ \emph {et~al.}(2011{\natexlab{b}})\citenamefont {Jafarizadeh}, \citenamefont {Sufiani},\ and\ \citenamefont {Mazhari~Khiavi}}]{Me11}%
  \BibitemOpen
  \bibfield  {author} {\bibinfo {author} {\bibfnamefont {M.~A.}\ \bibnamefont {Jafarizadeh}}, \bibinfo {author} {\bibfnamefont {R.}~\bibnamefont {Sufiani}}, \ and\ \bibinfo {author} {\bibfnamefont {Y.}~\bibnamefont {Mazhari~Khiavi}},\ }\bibfield  {title} {\enquote {\bibinfo {title} {Minimum error discrimination between similarity-transformed quantum states},}\ }\href {https://link.aps.org/doi/10.1103/PhysRevA.84.012102} {\bibfield  {journal} {\bibinfo  {journal} {Phys. Rev. A}\ }\textbf {\bibinfo {volume} {84}},\ \bibinfo {pages} {012102} (\bibinfo {year} {2011}{\natexlab{b}})}\BibitemShut {NoStop}%
\bibitem [{\citenamefont {Bae}(2013)}]{Me12}%
  \BibitemOpen
  \bibfield  {author} {\bibinfo {author} {\bibfnamefont {J.}~\bibnamefont {Bae}},\ }\bibfield  {title} {\enquote {\bibinfo {title} {Structure of minimum-error quantum state discrimination},}\ }\href {https://doi.org/10.1088/1367-2630/15/7/073037 Focus to learn more} {\bibfield  {journal} {\bibinfo  {journal} {New Journal of Physics}\ }\textbf {\bibinfo {volume} {15}},\ \bibinfo {pages} {073037} (\bibinfo {year} {2013})}\BibitemShut {NoStop}%
\bibitem [{\citenamefont {Sol\'{\i}s-Prosser}\ \emph {et~al.}(2017)\citenamefont {Sol\'{\i}s-Prosser}, \citenamefont {Fernandes}, \citenamefont {Jim\'enez}, \citenamefont {Delgado},\ and\ \citenamefont {Neves}}]{Me13}%
  \BibitemOpen
  \bibfield  {author} {\bibinfo {author} {\bibfnamefont {M.~A.}\ \bibnamefont {Sol\'{\i}s-Prosser}}, \bibinfo {author} {\bibfnamefont {M.~F.}\ \bibnamefont {Fernandes}}, \bibinfo {author} {\bibfnamefont {O.}~\bibnamefont {Jim\'enez}}, \bibinfo {author} {\bibfnamefont {A.}~\bibnamefont {Delgado}}, \ and\ \bibinfo {author} {\bibfnamefont {L.}~\bibnamefont {Neves}},\ }\bibfield  {title} {\enquote {\bibinfo {title} {Experimental minimum-error quantum-state discrimination in high dimensions},}\ }\href {https://link.aps.org/doi/10.1103/PhysRevLett.118.100501} {\bibfield  {journal} {\bibinfo  {journal} {Phys. Rev. Lett.}\ }\textbf {\bibinfo {volume} {118}},\ \bibinfo {pages} {100501} (\bibinfo {year} {2017})}\BibitemShut {NoStop}%
\bibitem [{\citenamefont {Loubenets}(2022{\natexlab{a}})}]{Me14}%
  \BibitemOpen
  \bibfield  {author} {\bibinfo {author} {\bibfnamefont {E.~R.}\ \bibnamefont {Loubenets}},\ }\bibfield  {title} {\enquote {\bibinfo {title} {General lower and upper bounds under minimum-error quantum state discrimination},}\ }\href {https://link.aps.org/doi/10.1103/PhysRevA.105.032410} {\bibfield  {journal} {\bibinfo  {journal} {Phys. Rev. A}\ }\textbf {\bibinfo {volume} {105}},\ \bibinfo {pages} {032410} (\bibinfo {year} {2022}{\natexlab{a}})}\BibitemShut {NoStop}%
\bibitem [{\citenamefont {Ha}\ and\ \citenamefont {Kim}(2022{\natexlab{a}})}]{Me15}%
  \BibitemOpen
  \bibfield  {author} {\bibinfo {author} {\bibfnamefont {D.}~\bibnamefont {Ha}}\ and\ \bibinfo {author} {\bibfnamefont {J.~S.}\ \bibnamefont {Kim}},\ }\bibfield  {title} {\enquote {\bibinfo {title} {Bound on local minimum-error discrimination of bipartite quantum states},}\ }\href {https://link.aps.org/doi/10.1103/PhysRevA.105.032421} {\bibfield  {journal} {\bibinfo  {journal} {Phys. Rev. A}\ }\textbf {\bibinfo {volume} {105}},\ \bibinfo {pages} {032421} (\bibinfo {year} {2022}{\natexlab{a}})}\BibitemShut {NoStop}%
\bibitem [{\citenamefont {Chefles}\ and\ \citenamefont {Barnett}(1998)}]{U2}%
  \BibitemOpen
  \bibfield  {author} {\bibinfo {author} {\bibfnamefont {A.}~\bibnamefont {Chefles}}\ and\ \bibinfo {author} {\bibfnamefont {S.~M.}\ \bibnamefont {Barnett}},\ }\bibfield  {title} {\enquote {\bibinfo {title} {Quantum state separation, unambiguous discrimination and exact cloning},}\ }\href {https://doi.org/10.1088/0305-4470/31/50/007} {\bibfield  {journal} {\bibinfo  {journal} {Journal of Physics A: Mathematical and General}\ }\textbf {\bibinfo {volume} {31}},\ \bibinfo {pages} {10097} (\bibinfo {year} {1998})}\BibitemShut {NoStop}%
\bibitem [{\citenamefont {Du\ifmmode~\check{s}\else \v{s}\fi{}ek}\ \emph {et~al.}(2000)\citenamefont {Du\ifmmode~\check{s}\else \v{s}\fi{}ek}, \citenamefont {Jahma},\ and\ \citenamefont {L\"utkenhaus}}]{U3}%
  \BibitemOpen
  \bibfield  {author} {\bibinfo {author} {\bibfnamefont {M.}~\bibnamefont {Du\ifmmode~\check{s}\else \v{s}\fi{}ek}}, \bibinfo {author} {\bibfnamefont {M.}~\bibnamefont {Jahma}}, \ and\ \bibinfo {author} {\bibfnamefont {N.}~\bibnamefont {L\"utkenhaus}},\ }\bibfield  {title} {\enquote {\bibinfo {title} {Unambiguous state discrimination in quantum cryptography with weak coherent states},}\ }\href {https://link.aps.org/doi/10.1103/PhysRevA.62.022306} {\bibfield  {journal} {\bibinfo  {journal} {Phys. Rev. A}\ }\textbf {\bibinfo {volume} {62}},\ \bibinfo {pages} {022306} (\bibinfo {year} {2000})}\BibitemShut {NoStop}%
\bibitem [{\citenamefont {Chefles}(2001)}]{U1}%
  \BibitemOpen
  \bibfield  {author} {\bibinfo {author} {\bibfnamefont {A.}~\bibnamefont {Chefles}},\ }\bibfield  {title} {\enquote {\bibinfo {title} {Unambiguous discrimination between linearly dependent states with multiple copies},}\ }\href {https://link.aps.org/doi/10.1103/PhysRevA.64.062305} {\bibfield  {journal} {\bibinfo  {journal} {Phys. Rev. A}\ }\textbf {\bibinfo {volume} {64}},\ \bibinfo {pages} {062305} (\bibinfo {year} {2001})}\BibitemShut {NoStop}%
\bibitem [{\citenamefont {Sun}\ \emph {et~al.}(2001)\citenamefont {Sun}, \citenamefont {Hillery},\ and\ \citenamefont {Bergou}}]{U4}%
  \BibitemOpen
  \bibfield  {author} {\bibinfo {author} {\bibfnamefont {Y.}~\bibnamefont {Sun}}, \bibinfo {author} {\bibfnamefont {M.}~\bibnamefont {Hillery}}, \ and\ \bibinfo {author} {\bibfnamefont {J.~A.}\ \bibnamefont {Bergou}},\ }\bibfield  {title} {\enquote {\bibinfo {title} {Optimum unambiguous discrimination between linearly independent nonorthogonal quantum states and its optical realization},}\ }\href {https://link.aps.org/doi/10.1103/PhysRevA.64.022311} {\bibfield  {journal} {\bibinfo  {journal} {Phys. Rev. A}\ }\textbf {\bibinfo {volume} {64}},\ \bibinfo {pages} {022311} (\bibinfo {year} {2001})}\BibitemShut {NoStop}%
\bibitem [{\citenamefont {Sun}\ \emph {et~al.}(2002)\citenamefont {Sun}, \citenamefont {Bergou},\ and\ \citenamefont {Hillery}}]{U5}%
  \BibitemOpen
  \bibfield  {author} {\bibinfo {author} {\bibfnamefont {Y.}~\bibnamefont {Sun}}, \bibinfo {author} {\bibfnamefont {J.~A.}\ \bibnamefont {Bergou}}, \ and\ \bibinfo {author} {\bibfnamefont {M.}~\bibnamefont {Hillery}},\ }\bibfield  {title} {\enquote {\bibinfo {title} {Optimum unambiguous discrimination between subsets of nonorthogonal quantum states},}\ }\href {https://link.aps.org/doi/10.1103/PhysRevA.66.032315} {\bibfield  {journal} {\bibinfo  {journal} {Phys. Rev. A}\ }\textbf {\bibinfo {volume} {66}},\ \bibinfo {pages} {032315} (\bibinfo {year} {2002})}\BibitemShut {NoStop}%
\bibitem [{\citenamefont {van Enk}(2002)}]{U7}%
  \BibitemOpen
  \bibfield  {author} {\bibinfo {author} {\bibfnamefont {S.~J.}\ \bibnamefont {van Enk}},\ }\bibfield  {title} {\enquote {\bibinfo {title} {Unambiguous state discrimination of coherent states with linear optics: Application to quantum cryptography},}\ }\href {https://link.aps.org/doi/10.1103/PhysRevA.66.042313} {\bibfield  {journal} {\bibinfo  {journal} {Phys. Rev. A}\ }\textbf {\bibinfo {volume} {66}},\ \bibinfo {pages} {042313} (\bibinfo {year} {2002})}\BibitemShut {NoStop}%
\bibitem [{\citenamefont {Eldar}(2003)}]{U6}%
  \BibitemOpen
  \bibfield  {author} {\bibinfo {author} {\bibfnamefont {Y.~C.}\ \bibnamefont {Eldar}},\ }\bibfield  {title} {\enquote {\bibinfo {title} {A semidefinite programming approach to optimal unambiguous discrimination of quantum states},}\ }\href {https://doi.org/10.1109/TIT.2002.807291} {\bibfield  {journal} {\bibinfo  {journal} {IEEE Trans. Inf. Theory}\ }\textbf {\bibinfo {volume} {49}},\ \bibinfo {pages} {446} (\bibinfo {year} {2003})}\BibitemShut {NoStop}%
\bibitem [{\citenamefont {Mohseni}\ \emph {et~al.}(2004{\natexlab{a}})\citenamefont {Mohseni}, \citenamefont {Steinberg},\ and\ \citenamefont {Bergou}}]{U8}%
  \BibitemOpen
  \bibfield  {author} {\bibinfo {author} {\bibfnamefont {M.}~\bibnamefont {Mohseni}}, \bibinfo {author} {\bibfnamefont {A.~M.}\ \bibnamefont {Steinberg}}, \ and\ \bibinfo {author} {\bibfnamefont {J.~A.}\ \bibnamefont {Bergou}},\ }\bibfield  {title} {\enquote {\bibinfo {title} {Optical realization of optimal unambiguous discrimination for pure and mixed quantum states},}\ }\href {https://link.aps.org/doi/10.1103/PhysRevLett.93.200403} {\bibfield  {journal} {\bibinfo  {journal} {Phys. Rev. Lett.}\ }\textbf {\bibinfo {volume} {93}},\ \bibinfo {pages} {200403} (\bibinfo {year} {2004}{\natexlab{a}})}\BibitemShut {NoStop}%
\bibitem [{\citenamefont {Feng}\ \emph {et~al.}(2004)\citenamefont {Feng}, \citenamefont {Duan},\ and\ \citenamefont {Ying}}]{U9}%
  \BibitemOpen
  \bibfield  {author} {\bibinfo {author} {\bibfnamefont {Y.}~\bibnamefont {Feng}}, \bibinfo {author} {\bibfnamefont {R.}~\bibnamefont {Duan}}, \ and\ \bibinfo {author} {\bibfnamefont {M.}~\bibnamefont {Ying}},\ }\bibfield  {title} {\enquote {\bibinfo {title} {Unambiguous discrimination between mixed quantum states},}\ }\href {https://link.aps.org/doi/10.1103/PhysRevA.70.012308} {\bibfield  {journal} {\bibinfo  {journal} {Phys. Rev. A}\ }\textbf {\bibinfo {volume} {70}},\ \bibinfo {pages} {012308} (\bibinfo {year} {2004})}\BibitemShut {NoStop}%
\bibitem [{\citenamefont {Herzog}\ and\ \citenamefont {Bergou}(2005)}]{U11}%
  \BibitemOpen
  \bibfield  {author} {\bibinfo {author} {\bibfnamefont {U.}~\bibnamefont {Herzog}}\ and\ \bibinfo {author} {\bibfnamefont {J.~A.}\ \bibnamefont {Bergou}},\ }\bibfield  {title} {\enquote {\bibinfo {title} {Optimum unambiguous discrimination of two mixed quantum states},}\ }\href {\doibase 10.1103/PhysRevA.71.050301} {\bibfield  {journal} {\bibinfo  {journal} {Phys. Rev. A}\ }\textbf {\bibinfo {volume} {71}},\ \bibinfo {pages} {050301} (\bibinfo {year} {2005})}\BibitemShut {NoStop}%
\bibitem [{\citenamefont {Bergou}\ and\ \citenamefont {Hillery}(2005)}]{U12}%
  \BibitemOpen
  \bibfield  {author} {\bibinfo {author} {\bibfnamefont {J.~A.}\ \bibnamefont {Bergou}}\ and\ \bibinfo {author} {\bibfnamefont {M.}~\bibnamefont {Hillery}},\ }\bibfield  {title} {\enquote {\bibinfo {title} {Universal programmable quantum state discriminator that is optimal for unambiguously distinguishing between unknown states},}\ }\href {https://link.aps.org/doi/10.1103/PhysRevLett.94.160501} {\bibfield  {journal} {\bibinfo  {journal} {Phys. Rev. Lett.}\ }\textbf {\bibinfo {volume} {94}},\ \bibinfo {pages} {160501} (\bibinfo {year} {2005})}\BibitemShut {NoStop}%
\bibitem [{\citenamefont {He}\ and\ \citenamefont {Bergou}(2006)}]{U13}%
  \BibitemOpen
  \bibfield  {author} {\bibinfo {author} {\bibfnamefont {B.}~\bibnamefont {He}}\ and\ \bibinfo {author} {\bibfnamefont {J.~A.}\ \bibnamefont {Bergou}},\ }\bibfield  {title} {\enquote {\bibinfo {title} {A general approach to physical realization of unambiguous quantum-state discrimination},}\ }\href {https://doi.org/10.1016/j.physleta.2006.03.076} {\bibfield  {journal} {\bibinfo  {journal} {Phys. Lett. A}\ }\textbf {\bibinfo {volume} {356}},\ \bibinfo {pages} {306} (\bibinfo {year} {2006})}\BibitemShut {NoStop}%
\bibitem [{\citenamefont {Jafarizadeh}\ \emph {et~al.}(2008)\citenamefont {Jafarizadeh}, \citenamefont {Rezaei}, \citenamefont {Karimi},\ and\ \citenamefont {Amiri}}]{U14}%
  \BibitemOpen
  \bibfield  {author} {\bibinfo {author} {\bibfnamefont {M.~A.}\ \bibnamefont {Jafarizadeh}}, \bibinfo {author} {\bibfnamefont {M.}~\bibnamefont {Rezaei}}, \bibinfo {author} {\bibfnamefont {N.}~\bibnamefont {Karimi}}, \ and\ \bibinfo {author} {\bibfnamefont {A.~R.}\ \bibnamefont {Amiri}},\ }\bibfield  {title} {\enquote {\bibinfo {title} {Optimal unambiguous discrimination of quantum states},}\ }\href {https://link.aps.org/doi/10.1103/PhysRevA.77.042314} {\bibfield  {journal} {\bibinfo  {journal} {Phys. Rev. A}\ }\textbf {\bibinfo {volume} {77}},\ \bibinfo {pages} {042314} (\bibinfo {year} {2008})}\BibitemShut {NoStop}%
\bibitem [{\citenamefont {Kleinmann}\ \emph {et~al.}(2010{\natexlab{a}})\citenamefont {Kleinmann}, \citenamefont {Kampermann},\ and\ \citenamefont {Bru{\ss}}}]{U15}%
  \BibitemOpen
  \bibfield  {author} {\bibinfo {author} {\bibfnamefont {M.}~\bibnamefont {Kleinmann}}, \bibinfo {author} {\bibfnamefont {H.}~\bibnamefont {Kampermann}}, \ and\ \bibinfo {author} {\bibfnamefont {D.}~\bibnamefont {Bru{\ss}}},\ }\bibfield  {title} {\enquote {\bibinfo {title} {Structural approach to unambiguous discrimination of two mixed quantum states},}\ }\href {https://doi.org/10.1063/1.3298683} {\bibfield  {journal} {\bibinfo  {journal} {J. Math. Phys.}\ }\textbf {\bibinfo {volume} {51}} (\bibinfo {year} {2010}{\natexlab{a}})}\BibitemShut {NoStop}%
\bibitem [{\citenamefont {Kleinmann}\ \emph {et~al.}(2010{\natexlab{b}})\citenamefont {Kleinmann}, \citenamefont {Kampermann},\ and\ \citenamefont {Bru\ss{}}}]{U10}%
  \BibitemOpen
  \bibfield  {author} {\bibinfo {author} {\bibfnamefont {M.}~\bibnamefont {Kleinmann}}, \bibinfo {author} {\bibfnamefont {H.}~\bibnamefont {Kampermann}}, \ and\ \bibinfo {author} {\bibfnamefont {D.}~\bibnamefont {Bru\ss{}}},\ }\bibfield  {title} {\enquote {\bibinfo {title} {Unambiguous discrimination of mixed quantum states: Optimal solution and case study},}\ }\href {https://link.aps.org/doi/10.1103/PhysRevA.81.020304} {\bibfield  {journal} {\bibinfo  {journal} {Phys. Rev. A}\ }\textbf {\bibinfo {volume} {81}},\ \bibinfo {pages} {020304} (\bibinfo {year} {2010}{\natexlab{b}})}\BibitemShut {NoStop}%
\bibitem [{\citenamefont {Bandyopadhyay}(2014)}]{U16}%
  \BibitemOpen
  \bibfield  {author} {\bibinfo {author} {\bibfnamefont {S.}~\bibnamefont {Bandyopadhyay}},\ }\bibfield  {title} {\enquote {\bibinfo {title} {Unambiguous discrimination of linearly independent pure quantum states: Optimal average probability of success},}\ }\href {https://link.aps.org/doi/10.1103/PhysRevA.90.030301} {\bibfield  {journal} {\bibinfo  {journal} {Phys. Rev. A}\ }\textbf {\bibinfo {volume} {90}},\ \bibinfo {pages} {030301} (\bibinfo {year} {2014})}\BibitemShut {NoStop}%
\bibitem [{\citenamefont {de~Assis}\ \emph {et~al.}(2017)\citenamefont {de~Assis}, \citenamefont {Sales},\ and\ \citenamefont {de~Almeida}}]{U17}%
  \BibitemOpen
  \bibfield  {author} {\bibinfo {author} {\bibfnamefont {R.~J.}\ \bibnamefont {de~Assis}}, \bibinfo {author} {\bibfnamefont {J.~S.}\ \bibnamefont {Sales}}, \ and\ \bibinfo {author} {\bibfnamefont {N.~G.}\ \bibnamefont {de~Almeida}},\ }\bibfield  {title} {\enquote {\bibinfo {title} {Unambiguous discrimination of nonorthogonal quantum states in cavity qed},}\ }\href {https://doi.org/10.1016/j.physleta.2017.07.017} {\bibfield  {journal} {\bibinfo  {journal} {Phys. Lett. A}\ }\textbf {\bibinfo {volume} {381}},\ \bibinfo {pages} {2927} (\bibinfo {year} {2017})}\BibitemShut {NoStop}%
\bibitem [{\citenamefont {Kim}\ \emph {et~al.}(2021)\citenamefont {Kim}, \citenamefont {Li}, \citenamefont {Kumar}, \citenamefont {Chunhe}, \citenamefont {Das}, \citenamefont {Sen}, \citenamefont {Pati},\ and\ \citenamefont {Wu}}]{U18}%
  \BibitemOpen
  \bibfield  {author} {\bibinfo {author} {\bibfnamefont {S.}~\bibnamefont {Kim}}, \bibinfo {author} {\bibfnamefont {L.}~\bibnamefont {Li}}, \bibinfo {author} {\bibfnamefont {A.}~\bibnamefont {Kumar}}, \bibinfo {author} {\bibfnamefont {X.}~\bibnamefont {Chunhe}}, \bibinfo {author} {\bibfnamefont {S.}~\bibnamefont {Das}}, \bibinfo {author} {\bibfnamefont {U.}~\bibnamefont {Sen}}, \bibinfo {author} {\bibfnamefont {A.}~\bibnamefont {Pati}}, \ and\ \bibinfo {author} {\bibfnamefont {J.}~\bibnamefont {Wu}},\ }\bibfield  {title} {\enquote {\bibinfo {title} {Protocol for unambiguous quantum state discrimination using quantum coherence},}\ }\href {https://doi.org/10.26421/QIC21.11-12-2} {\bibfield  {journal} {\bibinfo  {journal} {QIC}\ }\textbf {\bibinfo {volume} {21}},\ \bibinfo {pages} {931} (\bibinfo {year} {2021})}\BibitemShut {NoStop}%
\bibitem [{\citenamefont {Barnett}(1997)}]{Re1}%
  \BibitemOpen
  \bibfield  {author} {\bibinfo {author} {\bibfnamefont {S.~M.}\ \bibnamefont {Barnett}},\ }\bibfield  {title} {\enquote {\bibinfo {title} {Quantum information via novel measurements},}\ }\href {https://www.jstor.org/stable/54933} {\bibfield  {journal} {\bibinfo  {journal} {Philos. Trans. Royal Soc. A PHILOS T R SOC A}\ }\textbf {\bibinfo {volume} {355}},\ \bibinfo {pages} {2279} (\bibinfo {year} {1997})}\BibitemShut {NoStop}%
\bibitem [{\citenamefont {Chefles}(2000)}]{Re2}%
  \BibitemOpen
  \bibfield  {author} {\bibinfo {author} {\bibfnamefont {A.}~\bibnamefont {Chefles}},\ }\bibfield  {title} {\enquote {\bibinfo {title} {Quantum state discrimination},}\ }\href {https://doi.org/10.1080/00107510010002599} {\bibfield  {journal} {\bibinfo  {journal} {Contemp Phys}\ }\textbf {\bibinfo {volume} {41}},\ \bibinfo {pages} {401} (\bibinfo {year} {2000})}\BibitemShut {NoStop}%
\bibitem [{\citenamefont {Barnett}(2001)}]{Re3}%
  \BibitemOpen
  \bibfield  {author} {\bibinfo {author} {\bibfnamefont {S.~M.}\ \bibnamefont {Barnett}},\ }\bibfield  {title} {\enquote {\bibinfo {title} {Quantum limited state discrimination},}\ }\href {https://api.semanticscholar.org/CorpusID:85006425} {\bibfield  {journal} {\bibinfo  {journal} {Protein Sci}\ }\textbf {\bibinfo {volume} {49}},\ \bibinfo {pages} {909} (\bibinfo {year} {2001})}\BibitemShut {NoStop}%
\bibitem [{\citenamefont {Barnett}(2004)}]{Re4}%
  \BibitemOpen
  \bibfield  {author} {\bibinfo {author} {\bibfnamefont {S.~M.}\ \bibnamefont {Barnett}},\ }\bibfield  {title} {\enquote {\bibinfo {title} {Optical demonstrations of statistical decision theory for quantum systems},}\ }\href {https://api.semanticscholar.org/CorpusID:28903204} {\bibfield  {journal} {\bibinfo  {journal} {Quantum Inf. Comput.}\ }\textbf {\bibinfo {volume} {4}},\ \bibinfo {pages} {450} (\bibinfo {year} {2004})}\BibitemShut {NoStop}%
\bibitem [{\citenamefont {Paris}\ and\ \citenamefont {Rehacek}(2004)}]{Re5}%
  \BibitemOpen
  \bibfield  {author} {\bibinfo {author} {\bibfnamefont {M.}~\bibnamefont {Paris}}\ and\ \bibinfo {author} {\bibfnamefont {J.}~\bibnamefont {Rehacek}},\ }\href@noop {} {\emph {\bibinfo {title} {Quantum state estimation}}},\ Vol.\ \bibinfo {volume} {649}\ (\bibinfo  {publisher} {Springer Science \& Business Media},\ \bibinfo {year} {2004})\BibitemShut {NoStop}%
\bibitem [{\citenamefont {Bergou}(2007)}]{Re6}%
  \BibitemOpen
  \bibfield  {author} {\bibinfo {author} {\bibfnamefont {J.~A.}\ \bibnamefont {Bergou}},\ }\bibfield  {title} {\enquote {\bibinfo {title} {Quantum state discrimination and selected applications},}\ }in\ \href {https://api.semanticscholar.org/CorpusID:120041576} {\emph {\bibinfo {booktitle} {J. Phys. Conf. Ser}}},\ Vol.~\bibinfo {volume} {84}\ (\bibinfo {organization} {IOP Publishing},\ \bibinfo {year} {2007})\ p.\ \bibinfo {pages} {012001}\BibitemShut {NoStop}%
\bibitem [{\citenamefont {Barnett}\ and\ \citenamefont {Croke}(2009)}]{Re7}%
  \BibitemOpen
  \bibfield  {author} {\bibinfo {author} {\bibfnamefont {S.~M.}\ \bibnamefont {Barnett}}\ and\ \bibinfo {author} {\bibfnamefont {S.}~\bibnamefont {Croke}},\ }\bibfield  {title} {\enquote {\bibinfo {title} {Quantum state discrimination},}\ }\href {https://doi.org/10.48550/arXiv.0810.1970} {\bibfield  {journal} {\bibinfo  {journal} {Adv. Opt. Photonics}\ }\textbf {\bibinfo {volume} {1}},\ \bibinfo {pages} {238} (\bibinfo {year} {2009})}\BibitemShut {NoStop}%
\bibitem [{\citenamefont {Gungor}(2015)}]{Re8}%
  \BibitemOpen
  \bibfield  {author} {\bibinfo {author} {\bibfnamefont {O.}~\bibnamefont {Gungor}},\ }\bibfield  {title} {\enquote {\bibinfo {title} {Discrimination of quantum states under local operations and classical communication},}\ }\href {https://doi.org/10.48550/arXiv.1512.06707} {\bibfield  {journal} {\bibinfo  {journal} {arXiv preprint arXiv:1512.06707}\ } (\bibinfo {year} {2015})}\BibitemShut {NoStop}%
\bibitem [{\citenamefont {Bae}\ and\ \citenamefont {Kwek}(2015)}]{Re9}%
  \BibitemOpen
  \bibfield  {author} {\bibinfo {author} {\bibfnamefont {J.}~\bibnamefont {Bae}}\ and\ \bibinfo {author} {\bibfnamefont {L.-C.}\ \bibnamefont {Kwek}},\ }\bibfield  {title} {\enquote {\bibinfo {title} {Quantum state discrimination and its applications},}\ }\href {https://doi.org/10.1088/1751-8113/48/8/083001 Focus to learn more} {\bibfield  {journal} {\bibinfo  {journal} {J. Phys. A: Math. Theor.}\ }\textbf {\bibinfo {volume} {48}},\ \bibinfo {pages} {083001} (\bibinfo {year} {2015})}\BibitemShut {NoStop}%
\bibitem [{\citenamefont {Clarke}\ \emph {et~al.}(2001{\natexlab{a}})\citenamefont {Clarke}, \citenamefont {Chefles}, \citenamefont {Barnett},\ and\ \citenamefont {Riis}}]{E1}%
  \BibitemOpen
  \bibfield  {author} {\bibinfo {author} {\bibfnamefont {R.~B.~M.}\ \bibnamefont {Clarke}}, \bibinfo {author} {\bibfnamefont {A.}~\bibnamefont {Chefles}}, \bibinfo {author} {\bibfnamefont {S.~M.}\ \bibnamefont {Barnett}}, \ and\ \bibinfo {author} {\bibfnamefont {E.}~\bibnamefont {Riis}},\ }\bibfield  {title} {\enquote {\bibinfo {title} {Experimental demonstration of optimal unambiguous state discrimination},}\ }\href {https://link.aps.org/doi/10.1103/PhysRevA.63.040305} {\bibfield  {journal} {\bibinfo  {journal} {Phys. Rev. A}\ }\textbf {\bibinfo {volume} {63}},\ \bibinfo {pages} {040305} (\bibinfo {year} {2001}{\natexlab{a}})}\BibitemShut {NoStop}%
\bibitem [{\citenamefont {Clarke}\ \emph {et~al.}(2001{\natexlab{b}})\citenamefont {Clarke}, \citenamefont {Kendon}, \citenamefont {Chefles}, \citenamefont {Barnett}, \citenamefont {Riis},\ and\ \citenamefont {Sasaki}}]{E2}%
  \BibitemOpen
  \bibfield  {author} {\bibinfo {author} {\bibfnamefont {R.~B.~M.}\ \bibnamefont {Clarke}}, \bibinfo {author} {\bibfnamefont {V.~M.}\ \bibnamefont {Kendon}}, \bibinfo {author} {\bibfnamefont {A.}~\bibnamefont {Chefles}}, \bibinfo {author} {\bibfnamefont {S.~M.}\ \bibnamefont {Barnett}}, \bibinfo {author} {\bibfnamefont {E.}~\bibnamefont {Riis}}, \ and\ \bibinfo {author} {\bibfnamefont {M.}~\bibnamefont {Sasaki}},\ }\bibfield  {title} {\enquote {\bibinfo {title} {Experimental realization of optimal detection strategies for overcomplete states},}\ }\href {https://link.aps.org/doi/10.1103/PhysRevA.64.012303} {\bibfield  {journal} {\bibinfo  {journal} {Phys. Rev. A}\ }\textbf {\bibinfo {volume} {64}},\ \bibinfo {pages} {012303} (\bibinfo {year} {2001}{\natexlab{b}})}\BibitemShut {NoStop}%
\bibitem [{\citenamefont {Mohseni}\ \emph {et~al.}(2004{\natexlab{b}})\citenamefont {Mohseni}, \citenamefont {Steinberg},\ and\ \citenamefont {Bergou}}]{E3}%
  \BibitemOpen
  \bibfield  {author} {\bibinfo {author} {\bibfnamefont {M.}~\bibnamefont {Mohseni}}, \bibinfo {author} {\bibfnamefont {A.~M.}\ \bibnamefont {Steinberg}}, \ and\ \bibinfo {author} {\bibfnamefont {J.~A.}\ \bibnamefont {Bergou}},\ }\bibfield  {title} {\enquote {\bibinfo {title} {Optical realization of optimal unambiguous discrimination for pure and mixed quantum states},}\ }\href {https://link.aps.org/doi/10.1103/PhysRevLett.93.200403} {\bibfield  {journal} {\bibinfo  {journal} {Phys. Rev. Lett.}\ }\textbf {\bibinfo {volume} {93}},\ \bibinfo {pages} {200403} (\bibinfo {year} {2004}{\natexlab{b}})}\BibitemShut {NoStop}%
\bibitem [{\citenamefont {Mohseni}\ \emph {et~al.}(2004{\natexlab{c}})\citenamefont {Mohseni}, \citenamefont {Steinberg},\ and\ \citenamefont {Bergou}}]{E5}%
  \BibitemOpen
  \bibfield  {author} {\bibinfo {author} {\bibfnamefont {M.}~\bibnamefont {Mohseni}}, \bibinfo {author} {\bibfnamefont {Aephraim~M.}\ \bibnamefont {Steinberg}}, \ and\ \bibinfo {author} {\bibfnamefont {J.~A.}\ \bibnamefont {Bergou}},\ }\bibfield  {title} {\enquote {\bibinfo {title} {Optical realization of optimal unambiguous discrimination for pure and mixed quantum states},}\ }\href {https://link.aps.org/doi/10.1103/PhysRevLett.93.200403} {\bibfield  {journal} {\bibinfo  {journal} {Phys. Rev. Lett.}\ }\textbf {\bibinfo {volume} {93}},\ \bibinfo {pages} {200403} (\bibinfo {year} {2004}{\natexlab{c}})}\BibitemShut {NoStop}%
\bibitem [{\citenamefont {Lu}\ \emph {et~al.}(2010{\natexlab{b}})\citenamefont {Lu}, \citenamefont {Coish}, \citenamefont {Kaltenbaek}, \citenamefont {Hamel}, \citenamefont {Croke},\ and\ \citenamefont {Resch}}]{E6}%
  \BibitemOpen
  \bibfield  {author} {\bibinfo {author} {\bibfnamefont {Y.}~\bibnamefont {Lu}}, \bibinfo {author} {\bibfnamefont {N.}~\bibnamefont {Coish}}, \bibinfo {author} {\bibfnamefont {R.}~\bibnamefont {Kaltenbaek}}, \bibinfo {author} {\bibfnamefont {D.~R.}\ \bibnamefont {Hamel}}, \bibinfo {author} {\bibfnamefont {S.}~\bibnamefont {Croke}}, \ and\ \bibinfo {author} {\bibfnamefont {K.~J.}\ \bibnamefont {Resch}},\ }\bibfield  {title} {\enquote {\bibinfo {title} {Minimum-error discrimination of entangled quantum states},}\ }\href {https://link.aps.org/doi/10.1103/PhysRevA.82.042340} {\bibfield  {journal} {\bibinfo  {journal} {Phys. Rev. A}\ }\textbf {\bibinfo {volume} {82}},\ \bibinfo {pages} {042340} (\bibinfo {year} {2010}{\natexlab{b}})}\BibitemShut {NoStop}%
\bibitem [{\citenamefont {Sol\'{\i}s-Prosser}\ \emph {et~al.}(2016)\citenamefont {Sol\'{\i}s-Prosser}, \citenamefont {Gonz\'alez}, \citenamefont {Fuenzalida}, \citenamefont {G\'omez}, \citenamefont {Xavier}, \citenamefont {Delgado},\ and\ \citenamefont {Lima}}]{E4}%
  \BibitemOpen
  \bibfield  {author} {\bibinfo {author} {\bibfnamefont {M.~A.}\ \bibnamefont {Sol\'{\i}s-Prosser}}, \bibinfo {author} {\bibfnamefont {P.}~\bibnamefont {Gonz\'alez}}, \bibinfo {author} {\bibfnamefont {J.}~\bibnamefont {Fuenzalida}}, \bibinfo {author} {\bibfnamefont {S.}~\bibnamefont {G\'omez}}, \bibinfo {author} {\bibfnamefont {G.~B.}\ \bibnamefont {Xavier}}, \bibinfo {author} {\bibfnamefont {A.}~\bibnamefont {Delgado}}, \ and\ \bibinfo {author} {\bibfnamefont {G.}~\bibnamefont {Lima}},\ }\bibfield  {title} {\enquote {\bibinfo {title} {Experimental multiparty sequential state discrimination},}\ }\href {https://link.aps.org/doi/10.1103/PhysRevA.94.042309} {\bibfield  {journal} {\bibinfo  {journal} {Phys. Rev. A}\ }\textbf {\bibinfo {volume} {94}},\ \bibinfo {pages} {042309} (\bibinfo {year} {2016})}\BibitemShut {NoStop}%
\bibitem [{\citenamefont {Ferdinand}\ \emph {et~al.}(2017)\citenamefont {Ferdinand}, \citenamefont {DiMario},\ and\ \citenamefont {Becerra}}]{E7}%
  \BibitemOpen
  \bibfield  {author} {\bibinfo {author} {\bibfnamefont {A.R.}\ \bibnamefont {Ferdinand}}, \bibinfo {author} {\bibfnamefont {M.T.}\ \bibnamefont {DiMario}}, \ and\ \bibinfo {author} {\bibfnamefont {F.E.}\ \bibnamefont {Becerra}},\ }\bibfield  {title} {\enquote {\bibinfo {title} {Multi-state discrimination below the quantum noise limit at the single-photon level},}\ }\href {https://doi.org/10.1038/s41534-017-0042-2} {\bibfield  {journal} {\bibinfo  {journal} {Npj Quantum Inf}\ }\textbf {\bibinfo {volume} {3}},\ \bibinfo {pages} {43} (\bibinfo {year} {2017})}\BibitemShut {NoStop}%
\bibitem [{\citenamefont {Laneve}\ \emph {et~al.}(2022)\citenamefont {Laneve}, \citenamefont {Rota}, \citenamefont {Basset}, \citenamefont {Fiorente}, \citenamefont {Krieger}, \citenamefont {da~Silva}, \citenamefont {Buchinger}, \citenamefont {Stroj}, \citenamefont {Hoefling}, \citenamefont {Huber-Loyola} \emph {et~al.}}]{E8}%
  \BibitemOpen
  \bibfield  {author} {\bibinfo {author} {\bibfnamefont {Alessandro}\ \bibnamefont {Laneve}}, \bibinfo {author} {\bibfnamefont {Michele~B}\ \bibnamefont {Rota}}, \bibinfo {author} {\bibfnamefont {Francesco~Basso}\ \bibnamefont {Basset}}, \bibinfo {author} {\bibfnamefont {Nicola~P}\ \bibnamefont {Fiorente}}, \bibinfo {author} {\bibfnamefont {Tobias~M}\ \bibnamefont {Krieger}}, \bibinfo {author} {\bibfnamefont {Saimon F~Covre}\ \bibnamefont {da~Silva}}, \bibinfo {author} {\bibfnamefont {Quirin}\ \bibnamefont {Buchinger}}, \bibinfo {author} {\bibfnamefont {Sandra}\ \bibnamefont {Stroj}}, \bibinfo {author} {\bibfnamefont {Sven}\ \bibnamefont {Hoefling}}, \bibinfo {author} {\bibfnamefont {Tobias}\ \bibnamefont {Huber-Loyola}},  \emph {et~al.},\ }\bibfield  {title} {\enquote {\bibinfo {title} {Experimental multi-state quantum discrimination in the frequency domain with quantum dot light},}\ }\href {https://doi.org/10.48550/arXiv.2209.08324} {\bibfield  {journal} {\bibinfo  {journal} {arXiv preprint
  arXiv:2209.08324}\ } (\bibinfo {year} {2022})}\BibitemShut {NoStop}%
\bibitem [{\citenamefont {Becerra}\ \emph {et~al.}(2011)\citenamefont {Becerra}, \citenamefont {Fan}, \citenamefont {Baumgartner}, \citenamefont {Polyakov}, \citenamefont {Goldhar}, \citenamefont {Kosloski},\ and\ \citenamefont {Migdall}}]{E9}%
  \BibitemOpen
  \bibfield  {author} {\bibinfo {author} {\bibfnamefont {F.~E.}\ \bibnamefont {Becerra}}, \bibinfo {author} {\bibfnamefont {J.}~\bibnamefont {Fan}}, \bibinfo {author} {\bibfnamefont {G.}~\bibnamefont {Baumgartner}}, \bibinfo {author} {\bibfnamefont {S.~V.}\ \bibnamefont {Polyakov}}, \bibinfo {author} {\bibfnamefont {J.}~\bibnamefont {Goldhar}}, \bibinfo {author} {\bibfnamefont {J.~T.}\ \bibnamefont {Kosloski}}, \ and\ \bibinfo {author} {\bibfnamefont {A.}~\bibnamefont {Migdall}},\ }\bibfield  {title} {\enquote {\bibinfo {title} {$m$-ary-state phase-shift-keying discrimination below the homodyne limit},}\ }\href {https://link.aps.org/doi/10.1103/PhysRevA.84.062324} {\bibfield  {journal} {\bibinfo  {journal} {Phys. Rev. A}\ }\textbf {\bibinfo {volume} {84}},\ \bibinfo {pages} {062324} (\bibinfo {year} {2011})}\BibitemShut {NoStop}%
\bibitem [{\citenamefont {Becerra}\ \emph {et~al.}(2013{\natexlab{a}})\citenamefont {Becerra}, \citenamefont {Fan}, \citenamefont {Baumgartner}, \citenamefont {Goldhar}, \citenamefont {Kosloski},\ and\ \citenamefont {Migdall}}]{E10}%
  \BibitemOpen
  \bibfield  {author} {\bibinfo {author} {\bibfnamefont {F.~E.}\ \bibnamefont {Becerra}}, \bibinfo {author} {\bibfnamefont {Jingyun}\ \bibnamefont {Fan}}, \bibinfo {author} {\bibfnamefont {Gerald~B.}\ \bibnamefont {Baumgartner}}, \bibinfo {author} {\bibfnamefont {Julius}\ \bibnamefont {Goldhar}}, \bibinfo {author} {\bibfnamefont {J.~T.}\ \bibnamefont {Kosloski}}, \ and\ \bibinfo {author} {\bibfnamefont {Alan~L.}\ \bibnamefont {Migdall}},\ }\bibfield  {title} {\enquote {\bibinfo {title} {Experimental demonstration of a receiver beating the standard quantum limit for multiple nonorthogonal state discrimination},}\ }\href {https://api.semanticscholar.org/CorpusID:41194236} {\bibfield  {journal} {\bibinfo  {journal} {Nat. Photonics}\ }\textbf {\bibinfo {volume} {7}},\ \bibinfo {pages} {147} (\bibinfo {year} {2013}{\natexlab{a}})}\BibitemShut {NoStop}%
\bibitem [{\citenamefont {Becerra}\ \emph {et~al.}(2013{\natexlab{b}})\citenamefont {Becerra}, \citenamefont {Fan},\ and\ \citenamefont {Migdall}}]{E11}%
  \BibitemOpen
  \bibfield  {author} {\bibinfo {author} {\bibfnamefont {F.}~\bibnamefont {Becerra}}, \bibinfo {author} {\bibfnamefont {J.}~\bibnamefont {Fan}}, \ and\ \bibinfo {author} {\bibfnamefont {A.}~\bibnamefont {Migdall}},\ }\bibfield  {title} {\enquote {\bibinfo {title} {Implementation of generalized quantum measurements for unambiguous discrimination of multiple non-orthogonal coherent states},}\ }\href {\doibase 10.1038/ncomms3028} {\bibfield  {journal} {\bibinfo  {journal} {Nat. Commun.}\ }\textbf {\bibinfo {volume} {4}},\ \bibinfo {pages} {2028} (\bibinfo {year} {2013}{\natexlab{b}})}\BibitemShut {NoStop}%
\bibitem [{\citenamefont {Burenkov}\ \emph {et~al.}(2020)\citenamefont {Burenkov}, \citenamefont {Jabir}, \citenamefont {Battou},\ and\ \citenamefont {Polyakov}}]{E12}%
  \BibitemOpen
  \bibfield  {author} {\bibinfo {author} {\bibfnamefont {I.A.}\ \bibnamefont {Burenkov}}, \bibinfo {author} {\bibfnamefont {M.V.}\ \bibnamefont {Jabir}}, \bibinfo {author} {\bibfnamefont {A.}~\bibnamefont {Battou}}, \ and\ \bibinfo {author} {\bibfnamefont {S.V.}\ \bibnamefont {Polyakov}},\ }\bibfield  {title} {\enquote {\bibinfo {title} {Time-resolving quantum measurement enables energy-efficient, large-alphabet communication},}\ }\href {https://link.aps.org/doi/10.1103/PRXQuantum.1.010308} {\bibfield  {journal} {\bibinfo  {journal} {PRX Quantum}\ }\textbf {\bibinfo {volume} {1}},\ \bibinfo {pages} {010308} (\bibinfo {year} {2020})}\BibitemShut {NoStop}%
\bibitem [{\citenamefont {Sidhu}\ \emph {et~al.}(2021)\citenamefont {Sidhu}, \citenamefont {Izumi}, \citenamefont {Neergaard-Nielsen}, \citenamefont {Lupo},\ and\ \citenamefont {Andersen}}]{E13}%
  \BibitemOpen
  \bibfield  {author} {\bibinfo {author} {\bibfnamefont {J.~S.}\ \bibnamefont {Sidhu}}, \bibinfo {author} {\bibfnamefont {S.}~\bibnamefont {Izumi}}, \bibinfo {author} {\bibfnamefont {J.~S.}\ \bibnamefont {Neergaard-Nielsen}}, \bibinfo {author} {\bibfnamefont {C.}~\bibnamefont {Lupo}}, \ and\ \bibinfo {author} {\bibfnamefont {U.~L.}\ \bibnamefont {Andersen}},\ }\bibfield  {title} {\enquote {\bibinfo {title} {Quantum receiver for phase-shift keying at the single-photon level},}\ }\href {https://link.aps.org/doi/10.1103/PRXQuantum.2.010332} {\bibfield  {journal} {\bibinfo  {journal} {PRX Quantum}\ }\textbf {\bibinfo {volume} {2}},\ \bibinfo {pages} {010332} (\bibinfo {year} {2021})}\BibitemShut {NoStop}%
\bibitem [{\citenamefont {Ac\'{\i}n}\ \emph {et~al.}(2005)\citenamefont {Ac\'{\i}n}, \citenamefont {Bagan}, \citenamefont {Baig}, \citenamefont {Masanes},\ and\ \citenamefont {Mu\~noz Tapia}}]{Mc}%
  \BibitemOpen
  \bibfield  {author} {\bibinfo {author} {\bibfnamefont {A.}~\bibnamefont {Ac\'{\i}n}}, \bibinfo {author} {\bibfnamefont {E.}~\bibnamefont {Bagan}}, \bibinfo {author} {\bibfnamefont {M.}~\bibnamefont {Baig}}, \bibinfo {author} {\bibfnamefont {Ll.}\ \bibnamefont {Masanes}}, \ and\ \bibinfo {author} {\bibfnamefont {R.}~\bibnamefont {Mu\~noz Tapia}},\ }\bibfield  {title} {\enquote {\bibinfo {title} {Multiple-copy two-state discrimination with individual measurements},}\ }\href {https://link.aps.org/doi/10.1103/PhysRevA.71.032338} {\bibfield  {journal} {\bibinfo  {journal} {Phys. Rev. A}\ }\textbf {\bibinfo {volume} {71}},\ \bibinfo {pages} {032338} (\bibinfo {year} {2005})}\BibitemShut {NoStop}%
\bibitem [{\citenamefont {Assalini}\ \emph {et~al.}(2011)\citenamefont {Assalini}, \citenamefont {Dalla~P.},\ and\ \citenamefont {Pierobon}}]{Mc2}%
  \BibitemOpen
  \bibfield  {author} {\bibinfo {author} {\bibfnamefont {A.}~\bibnamefont {Assalini}}, \bibinfo {author} {\bibfnamefont {Nicola}\ \bibnamefont {Dalla~P.}}, \ and\ \bibinfo {author} {\bibfnamefont {G.}~\bibnamefont {Pierobon}},\ }\bibfield  {title} {\enquote {\bibinfo {title} {Revisiting the dolinar receiver through multiple-copy state discrimination theory},}\ }\href {https://link.aps.org/doi/10.1103/PhysRevA.84.022342} {\bibfield  {journal} {\bibinfo  {journal} {Phys. Rev. A}\ }\textbf {\bibinfo {volume} {84}},\ \bibinfo {pages} {022342} (\bibinfo {year} {2011})}\BibitemShut {NoStop}%
\bibitem [{\citenamefont {Slussarenko}\ \emph {et~al.}(2017)\citenamefont {Slussarenko}, \citenamefont {Weston}, \citenamefont {Li}, \citenamefont {Campbell}, \citenamefont {Wiseman},\ and\ \citenamefont {Pryde}}]{Mc4}%
  \BibitemOpen
  \bibfield  {author} {\bibinfo {author} {\bibfnamefont {S.}~\bibnamefont {Slussarenko}}, \bibinfo {author} {\bibfnamefont {M.~M.}\ \bibnamefont {Weston}}, \bibinfo {author} {\bibfnamefont {J.-G.}\ \bibnamefont {Li}}, \bibinfo {author} {\bibfnamefont {N.}~\bibnamefont {Campbell}}, \bibinfo {author} {\bibfnamefont {H.~M.}\ \bibnamefont {Wiseman}}, \ and\ \bibinfo {author} {\bibfnamefont {G.~J.}\ \bibnamefont {Pryde}},\ }\bibfield  {title} {\enquote {\bibinfo {title} {Quantum state discrimination using the minimum average number of copies},}\ }\href {https://link.aps.org/doi/10.1103/PhysRevLett.118.030502} {\bibfield  {journal} {\bibinfo  {journal} {Phys. Rev. Lett.}\ }\textbf {\bibinfo {volume} {118}},\ \bibinfo {pages} {030502} (\bibinfo {year} {2017})}\BibitemShut {NoStop}%
\bibitem [{\citenamefont {Ogawa}\ and\ \citenamefont {Nagaoka}(1999)}]{Up1}%
  \BibitemOpen
  \bibfield  {author} {\bibinfo {author} {\bibfnamefont {T.}~\bibnamefont {Ogawa}}\ and\ \bibinfo {author} {\bibfnamefont {H.}~\bibnamefont {Nagaoka}},\ }\bibfield  {title} {\enquote {\bibinfo {title} {Strong converse and stein's lemma in quantum hypothesis testing},}\ }\href {https://api.semanticscholar.org/CorpusID:16725543} {\bibfield  {journal} {\bibinfo  {journal} {IEEE Trans. Inf. Theory}\ }\textbf {\bibinfo {volume} {46}},\ \bibinfo {pages} {2428} (\bibinfo {year} {1999})}\BibitemShut {NoStop}%
\bibitem [{\citenamefont {Nayak}\ and\ \citenamefont {Salzman}(2006)}]{Up2}%
  \BibitemOpen
  \bibfield  {author} {\bibinfo {author} {\bibfnamefont {A.}~\bibnamefont {Nayak}}\ and\ \bibinfo {author} {\bibfnamefont {J.}~\bibnamefont {Salzman}},\ }\bibfield  {title} {\enquote {\bibinfo {title} {Limits on the ability of quantum states to convey classical messages},}\ }\href {https://api.semanticscholar.org/CorpusID:13475632} {\bibfield  {journal} {\bibinfo  {journal} {J. ACM}\ }\textbf {\bibinfo {volume} {53}},\ \bibinfo {pages} {184} (\bibinfo {year} {2006})}\BibitemShut {NoStop}%
\bibitem [{\citenamefont {Montanaro}(2007)}]{Up3}%
  \BibitemOpen
  \bibfield  {author} {\bibinfo {author} {\bibfnamefont {A.}~\bibnamefont {Montanaro}},\ }\bibfield  {title} {\enquote {\bibinfo {title} {On the distinguishability of random quantum states},}\ }\href {https://doi.org/10.1007/s00220-007-0221-7} {\bibfield  {journal} {\bibinfo  {journal} {Comm. Math. Phys.}\ }\textbf {\bibinfo {volume} {273}},\ \bibinfo {pages} {619} (\bibinfo {year} {2007})}\BibitemShut {NoStop}%
\bibitem [{\citenamefont {Montanaro}(2008)}]{Up4}%
  \BibitemOpen
  \bibfield  {author} {\bibinfo {author} {\bibfnamefont {A.}~\bibnamefont {Montanaro}},\ }\bibfield  {title} {\enquote {\bibinfo {title} {A lower bound on the probability of error in quantum state discrimination},}\ }in\ \href {https://api.semanticscholar.org/CorpusID:14392052} {\emph {\bibinfo {booktitle} {2008 IEEE Information Theory Workshop}}}\ (\bibinfo {organization} {IEEE},\ \bibinfo {year} {2008})\ p.\ \bibinfo {pages} {378}\BibitemShut {NoStop}%
\bibitem [{\citenamefont {Hayashi}\ \emph {et~al.}(2008)\citenamefont {Hayashi}, \citenamefont {Kawachi},\ and\ \citenamefont {Kobayashi}}]{Up5}%
  \BibitemOpen
  \bibfield  {author} {\bibinfo {author} {\bibfnamefont {M.}~\bibnamefont {Hayashi}}, \bibinfo {author} {\bibfnamefont {A.}~\bibnamefont {Kawachi}}, \ and\ \bibinfo {author} {\bibfnamefont {H.}~\bibnamefont {Kobayashi}},\ }\bibfield  {title} {\enquote {\bibinfo {title} {Quantum measurements for hidden subgroup problems with optimal sample complexity},}\ }\href {https://api.semanticscholar.org/CorpusID:16892644} {\bibfield  {journal} {\bibinfo  {journal} {Quantum Inf. Comput.}\ }\textbf {\bibinfo {volume} {8}},\ \bibinfo {pages} {345} (\bibinfo {year} {2008})}\BibitemShut {NoStop}%
\bibitem [{\citenamefont {Qiu}(2008{\natexlab{b}})}]{Up6}%
  \BibitemOpen
  \bibfield  {author} {\bibinfo {author} {\bibfnamefont {D.}~\bibnamefont {Qiu}},\ }\bibfield  {title} {\enquote {\bibinfo {title} {Minimum-error discrimination between mixed quantum states},}\ }\href {https://link.aps.org/doi/10.1103/PhysRevA.77.012328} {\bibfield  {journal} {\bibinfo  {journal} {Phys. Rev. A}\ }\textbf {\bibinfo {volume} {77}},\ \bibinfo {pages} {012328} (\bibinfo {year} {2008}{\natexlab{b}})}\BibitemShut {NoStop}%
\bibitem [{\citenamefont {Tyson}(2009)}]{Up7}%
  \BibitemOpen
  \bibfield  {author} {\bibinfo {author} {\bibfnamefont {J.}~\bibnamefont {Tyson}},\ }\bibfield  {title} {\enquote {\bibinfo {title} {Minimum-error quantum distinguishability bounds from matrix monotone functions: A comment on “two-sided estimates of minimum-error distinguishability of mixed quantum states via generalized holevo–curlander bounds”},}\ }\href {https://api.semanticscholar.org/CorpusID:120982971} {\bibfield  {journal} {\bibinfo  {journal} {J. Math. Phys.}\ }\textbf {\bibinfo {volume} {50}},\ \bibinfo {pages} {062102} (\bibinfo {year} {2009})}\BibitemShut {NoStop}%
\bibitem [{\citenamefont {Hwang}\ and\ \citenamefont {Bae}(2010)}]{Up8}%
  \BibitemOpen
  \bibfield  {author} {\bibinfo {author} {\bibfnamefont {W.-Y.}\ \bibnamefont {Hwang}}\ and\ \bibinfo {author} {\bibfnamefont {J.}~\bibnamefont {Bae}},\ }\bibfield  {title} {\enquote {\bibinfo {title} {Minimum-error state discrimination constrained by the no-signaling principle},}\ }\href {https://api.semanticscholar.org/CorpusID:119164785} {\bibfield  {journal} {\bibinfo  {journal} {J. Math. Phys.}\ }\textbf {\bibinfo {volume} {51}} (\bibinfo {year} {2010})}\BibitemShut {NoStop}%
\bibitem [{\citenamefont {Qiu}\ and\ \citenamefont {Li}(2010{\natexlab{b}})}]{Up9}%
  \BibitemOpen
  \bibfield  {author} {\bibinfo {author} {\bibfnamefont {D.}~\bibnamefont {Qiu}}\ and\ \bibinfo {author} {\bibfnamefont {L.}~\bibnamefont {Li}},\ }\bibfield  {title} {\enquote {\bibinfo {title} {Minimum-error discrimination of quantum states: Bounds and comparisons},}\ }\href {https://link.aps.org/doi/10.1103/PhysRevA.81.042329} {\bibfield  {journal} {\bibinfo  {journal} {Phys. Rev. A}\ }\textbf {\bibinfo {volume} {81}},\ \bibinfo {pages} {042329} (\bibinfo {year} {2010}{\natexlab{b}})}\BibitemShut {NoStop}%
\bibitem [{\citenamefont {Audenaert}\ and\ \citenamefont {Mosonyi}(2014)}]{Upb9}%
  \BibitemOpen
  \bibfield  {author} {\bibinfo {author} {\bibfnamefont {K.~M.~R.}\ \bibnamefont {Audenaert}}\ and\ \bibinfo {author} {\bibfnamefont {M.}~\bibnamefont {Mosonyi}},\ }\bibfield  {title} {\enquote {\bibinfo {title} {Upper bounds on the error probabilities and asymptotic error exponents in quantum multiple state discrimination},}\ }\href {https://api.semanticscholar.org/CorpusID:1811103} {\bibfield  {journal} {\bibinfo  {journal} {J. Math. Phys}\ }\textbf {\bibinfo {volume} {55}},\ \bibinfo {pages} {102201} (\bibinfo {year} {2014})}\BibitemShut {NoStop}%
\bibitem [{\citenamefont {Bagan}\ \emph {et~al.}(2016)\citenamefont {Bagan}, \citenamefont {Bergou}, \citenamefont {Cottrell},\ and\ \citenamefont {Hillery}}]{Up10}%
  \BibitemOpen
  \bibfield  {author} {\bibinfo {author} {\bibfnamefont {E.}~\bibnamefont {Bagan}}, \bibinfo {author} {\bibfnamefont {J.~A.}\ \bibnamefont {Bergou}}, \bibinfo {author} {\bibfnamefont {S.~S.}\ \bibnamefont {Cottrell}}, \ and\ \bibinfo {author} {\bibfnamefont {M.}~\bibnamefont {Hillery}},\ }\bibfield  {title} {\enquote {\bibinfo {title} {Relations between coherence and path information},}\ }\href {https://link.aps.org/doi/10.1103/PhysRevLett.116.160406} {\bibfield  {journal} {\bibinfo  {journal} {Phys. Rev. Lett.}\ }\textbf {\bibinfo {volume} {116}},\ \bibinfo {pages} {160406} (\bibinfo {year} {2016})}\BibitemShut {NoStop}%
\bibitem [{\citenamefont {Nakahira}\ \emph {et~al.}(2018)\citenamefont {Nakahira}, \citenamefont {Kato},\ and\ \citenamefont {Usuda}}]{Upb11}%
  \BibitemOpen
  \bibfield  {author} {\bibinfo {author} {\bibfnamefont {K.}~\bibnamefont {Nakahira}}, \bibinfo {author} {\bibfnamefont {K.}~\bibnamefont {Kato}}, \ and\ \bibinfo {author} {\bibfnamefont {T.~S.}\ \bibnamefont {Usuda}},\ }\bibfield  {title} {\enquote {\bibinfo {title} {Optimal discrimination of optical coherent states cannot always be realized by interfering with coherent light, photon counting, and feedback},}\ }\href {https://link.aps.org/doi/10.1103/PhysRevA.97.022320} {\bibfield  {journal} {\bibinfo  {journal} {Phys. Rev. A}\ }\textbf {\bibinfo {volume} {97}},\ \bibinfo {pages} {022320} (\bibinfo {year} {2018})}\BibitemShut {NoStop}%
\bibitem [{\citenamefont {Loubenets}(2022{\natexlab{b}})}]{Up11}%
  \BibitemOpen
  \bibfield  {author} {\bibinfo {author} {\bibfnamefont {E.~R.}\ \bibnamefont {Loubenets}},\ }\bibfield  {title} {\enquote {\bibinfo {title} {General lower and upper bounds under minimum-error quantum state discrimination},}\ }\href {https://link.aps.org/doi/10.1103/PhysRevA.105.032410} {\bibfield  {journal} {\bibinfo  {journal} {Phys. Rev. A}\ }\textbf {\bibinfo {volume} {105}},\ \bibinfo {pages} {032410} (\bibinfo {year} {2022}{\natexlab{b}})}\BibitemShut {NoStop}%
\bibitem [{\citenamefont {Ha}\ and\ \citenamefont {Kim}(2022{\natexlab{b}})}]{UL}%
  \BibitemOpen
  \bibfield  {author} {\bibinfo {author} {\bibfnamefont {D.}~\bibnamefont {Ha}}\ and\ \bibinfo {author} {\bibfnamefont {J.~S.}\ \bibnamefont {Kim}},\ }\bibfield  {title} {\enquote {\bibinfo {title} {Bound on local minimum-error discrimination of bipartite quantum states},}\ }\href {https://link.aps.org/doi/10.1103/PhysRevA.105.032421} {\bibfield  {journal} {\bibinfo  {journal} {Phys. Rev. A}\ }\textbf {\bibinfo {volume} {105}},\ \bibinfo {pages} {032421} (\bibinfo {year} {2022}{\natexlab{b}})}\BibitemShut {NoStop}%
\bibitem [{\citenamefont {Skrzypczyk}\ \emph {et~al.}(2019)\citenamefont {Skrzypczyk}, \citenamefont {\ifmmode \check{S}\else \v{S}\fi{}upi\ifmmode~\acute{c}\else \'{c}\fi{}},\ and\ \citenamefont {Cavalcanti}}]{Co}%
  \BibitemOpen
  \bibfield  {author} {\bibinfo {author} {\bibfnamefont {P.}~\bibnamefont {Skrzypczyk}}, \bibinfo {author} {\bibfnamefont {I.}~\bibnamefont {\ifmmode \check{S}\else \v{S}\fi{}upi\ifmmode~\acute{c}\else \'{c}\fi{}}}, \ and\ \bibinfo {author} {\bibfnamefont {D.}~\bibnamefont {Cavalcanti}},\ }\bibfield  {title} {\enquote {\bibinfo {title} {All sets of incompatible measurements give an advantage in quantum state discrimination},}\ }\href {https://link.aps.org/doi/10.1103/PhysRevLett.122.130403} {\bibfield  {journal} {\bibinfo  {journal} {Phys. Rev. Lett.}\ }\textbf {\bibinfo {volume} {122}},\ \bibinfo {pages} {130403} (\bibinfo {year} {2019})}\BibitemShut {NoStop}%
\bibitem [{\citenamefont {L\"u}(2021)}]{Cb1}%
  \BibitemOpen
  \bibfield  {author} {\bibinfo {author} {\bibfnamefont {X.}~\bibnamefont {L\"u}},\ }\bibfield  {title} {\enquote {\bibinfo {title} {Lower bounds on the failure probability of unambiguous discrimination},}\ }\href {https://link.aps.org/doi/10.1103/PhysRevA.103.022216} {\bibfield  {journal} {\bibinfo  {journal} {Phys. Rev. A}\ }\textbf {\bibinfo {volume} {103}},\ \bibinfo {pages} {022216} (\bibinfo {year} {2021})}\BibitemShut {NoStop}%
\bibitem [{\citenamefont {Lü}(2022)}]{Cb2}%
  \BibitemOpen
  \bibfield  {author} {\bibinfo {author} {\bibfnamefont {X.}~\bibnamefont {Lü}},\ }\bibfield  {title} {\enquote {\bibinfo {title} {An upper bound on the success probability of minimum-error discrimination},}\ }\href {https://www.researchgate.net/publication/363877616_An_upper_bound_on_the_success_probability_of_minimum-error_discrimination} {\bibfield  {journal} {\bibinfo  {journal} {Phys. Lett. A.}\ }\textbf {\bibinfo {volume} {452}},\ \bibinfo {pages} {128461} (\bibinfo {year} {2022})}\BibitemShut {NoStop}%
\bibitem [{\citenamefont {Vidal}\ and\ \citenamefont {Tarrach}(1999)}]{Ref1}%
  \BibitemOpen
  \bibfield  {author} {\bibinfo {author} {\bibfnamefont {G.}~\bibnamefont {Vidal}}\ and\ \bibinfo {author} {\bibfnamefont {R.}~\bibnamefont {Tarrach}},\ }\bibfield  {title} {\enquote {\bibinfo {title} {Robustness of entanglement},}\ }\href {https://link.aps.org/doi/10.1103/PhysRevA.59.141} {\bibfield  {journal} {\bibinfo  {journal} {Phys. Rev. A}\ }\textbf {\bibinfo {volume} {59}},\ \bibinfo {pages} {141} (\bibinfo {year} {1999})}\BibitemShut {NoStop}%
\bibitem [{\citenamefont {\ifmmode~\dot{Z}\else \.{Z}\fi{}yczkowski}\ \emph {et~al.}(1998)\citenamefont {\ifmmode~\dot{Z}\else \.{Z}\fi{}yczkowski}, \citenamefont {Horodecki}, \citenamefont {Sanpera},\ and\ \citenamefont {Lewenstein}}]{SB1}%
  \BibitemOpen
  \bibfield  {author} {\bibinfo {author} {\bibfnamefont {K.}~\bibnamefont {\ifmmode~\dot{Z}\else \.{Z}\fi{}yczkowski}}, \bibinfo {author} {\bibfnamefont {P.}~\bibnamefont {Horodecki}}, \bibinfo {author} {\bibfnamefont {A.}~\bibnamefont {Sanpera}}, \ and\ \bibinfo {author} {\bibfnamefont {M.}~\bibnamefont {Lewenstein}},\ }\bibfield  {title} {\enquote {\bibinfo {title} {Volume of the set of separable states},}\ }\href {https://link.aps.org/doi/10.1103/PhysRevA.58.883} {\bibfield  {journal} {\bibinfo  {journal} {Phys. Rev. A}\ }\textbf {\bibinfo {volume} {58}},\ \bibinfo {pages} {883} (\bibinfo {year} {1998})}\BibitemShut {NoStop}%
\bibitem [{\citenamefont {Braunstein}\ \emph {et~al.}(1999)\citenamefont {Braunstein}, \citenamefont {Caves}, \citenamefont {Jozsa}, \citenamefont {Linden}, \citenamefont {Popescu},\ and\ \citenamefont {Schack}}]{SB2}%
  \BibitemOpen
  \bibfield  {author} {\bibinfo {author} {\bibfnamefont {S.~L.}\ \bibnamefont {Braunstein}}, \bibinfo {author} {\bibfnamefont {C.~M.}\ \bibnamefont {Caves}}, \bibinfo {author} {\bibfnamefont {R.}~\bibnamefont {Jozsa}}, \bibinfo {author} {\bibfnamefont {N.}~\bibnamefont {Linden}}, \bibinfo {author} {\bibfnamefont {S.}~\bibnamefont {Popescu}}, \ and\ \bibinfo {author} {\bibfnamefont {R.}~\bibnamefont {Schack}},\ }\bibfield  {title} {\enquote {\bibinfo {title} {Separability of very noisy mixed states and implications for nmr quantum computing},}\ }\href {https://link.aps.org/doi/10.1103/PhysRevLett.83.1054} {\bibfield  {journal} {\bibinfo  {journal} {Phys. Rev. Lett.}\ }\textbf {\bibinfo {volume} {83}},\ \bibinfo {pages} {1054} (\bibinfo {year} {1999})}\BibitemShut {NoStop}%
\bibitem [{\citenamefont {Szarek}(2005)}]{SB3}%
  \BibitemOpen
  \bibfield  {author} {\bibinfo {author} {\bibfnamefont {S.~J.}\ \bibnamefont {Szarek}},\ }\bibfield  {title} {\enquote {\bibinfo {title} {Volume of separable states is super-doubly-exponentially small in the number of qubits},}\ }\href {https://link.aps.org/doi/10.1103/PhysRevA.72.032304} {\bibfield  {journal} {\bibinfo  {journal} {Phys. Rev. A}\ }\textbf {\bibinfo {volume} {72}},\ \bibinfo {pages} {032304} (\bibinfo {year} {2005})}\BibitemShut {NoStop}%
\bibitem [{\citenamefont {Walgate}\ \emph {et~al.}(2000)\citenamefont {Walgate}, \citenamefont {Short}, \citenamefont {Hardy},\ and\ \citenamefont {Vedral}}]{jon}%
  \BibitemOpen
  \bibfield  {author} {\bibinfo {author} {\bibfnamefont {J.}~\bibnamefont {Walgate}}, \bibinfo {author} {\bibfnamefont {A.~J.}\ \bibnamefont {Short}}, \bibinfo {author} {\bibfnamefont {L.}~\bibnamefont {Hardy}}, \ and\ \bibinfo {author} {\bibfnamefont {V.}~\bibnamefont {Vedral}},\ }\bibfield  {title} {\enquote {\bibinfo {title} {Local distinguishability of multipartite orthogonal quantum states},}\ }\href {https://link.aps.org/doi/10.1103/PhysRevLett.85.4972} {\bibfield  {journal} {\bibinfo  {journal} {Phys. Rev. Lett.}\ }\textbf {\bibinfo {volume} {85}},\ \bibinfo {pages} {4972} (\bibinfo {year} {2000})}\BibitemShut {NoStop}%
\bibitem [{\citenamefont {Bennett}\ \emph {et~al.}(1999)\citenamefont {Bennett}, \citenamefont {DiVincenzo}, \citenamefont {Fuchs}, \citenamefont {Mor}, \citenamefont {Rains}, \citenamefont {Shor}, \citenamefont {Smolin},\ and\ \citenamefont {Wootters}}]{Ben}%
  \BibitemOpen
  \bibfield  {author} {\bibinfo {author} {\bibfnamefont {C.~H.}\ \bibnamefont {Bennett}}, \bibinfo {author} {\bibfnamefont {D.~P.}\ \bibnamefont {DiVincenzo}}, \bibinfo {author} {\bibfnamefont {C.~A.}\ \bibnamefont {Fuchs}}, \bibinfo {author} {\bibfnamefont {T.}~\bibnamefont {Mor}}, \bibinfo {author} {\bibfnamefont {E.}~\bibnamefont {Rains}}, \bibinfo {author} {\bibfnamefont {P.~W.}\ \bibnamefont {Shor}}, \bibinfo {author} {\bibfnamefont {J.~A.}\ \bibnamefont {Smolin}}, \ and\ \bibinfo {author} {\bibfnamefont {W.~K.}\ \bibnamefont {Wootters}},\ }\bibfield  {title} {\enquote {\bibinfo {title} {Quantum nonlocality without entanglement},}\ }\href {https://link.aps.org/doi/10.1103/PhysRevA.59.1070} {\bibfield  {journal} {\bibinfo  {journal} {Phys. Rev. A}\ }\textbf {\bibinfo {volume} {59}},\ \bibinfo {pages} {1070} (\bibinfo {year} {1999})}\BibitemShut {NoStop}%
\bibitem [{\citenamefont {Horodecki}\ \emph {et~al.}(2003)\citenamefont {Horodecki}, \citenamefont {Sen(De)}, \citenamefont {Sen},\ and\ \citenamefont {Horodecki}}]{NLWE}%
  \BibitemOpen
  \bibfield  {author} {\bibinfo {author} {\bibfnamefont {M.}~\bibnamefont {Horodecki}}, \bibinfo {author} {\bibfnamefont {A.}~\bibnamefont {Sen(De)}}, \bibinfo {author} {\bibfnamefont {U.}~\bibnamefont {Sen}}, \ and\ \bibinfo {author} {\bibfnamefont {K.}~\bibnamefont {Horodecki}},\ }\bibfield  {title} {\enquote {\bibinfo {title} {Local indistinguishability: More nonlocality with less entanglement},}\ }\href {https://link.aps.org/doi/10.1103/PhysRevLett.90.047902} {\bibfield  {journal} {\bibinfo  {journal} {Phys. Rev. Lett.}\ }\textbf {\bibinfo {volume} {90}},\ \bibinfo {pages} {047902} (\bibinfo {year} {2003})}\BibitemShut {NoStop}%
\bibitem [{\citenamefont {Simon}(2000)}]{ppt}%
  \BibitemOpen
  \bibfield  {author} {\bibinfo {author} {\bibfnamefont {R.}~\bibnamefont {Simon}},\ }\bibfield  {title} {\enquote {\bibinfo {title} {Peres-horodecki separability criterion for continuous variable systems},}\ }\href {https://link.aps.org/doi/10.1103/PhysRevLett.84.2726} {\bibfield  {journal} {\bibinfo  {journal} {Phys. Rev. Lett.}\ }\textbf {\bibinfo {volume} {84}},\ \bibinfo {pages} {2726} (\bibinfo {year} {2000})}\BibitemShut {NoStop}%
\end{thebibliography}%
\end{document}